\documentclass[review]{elsarticle}

\usepackage{lineno,hyperref}
\usepackage{mathtools}
\usepackage{amssymb}
\modulolinenumbers[5]

\usepackage[utf8]{inputenc}
\usepackage[english]{babel}
\usepackage{amsmath}
\usepackage{dsfont}
\usepackage{mathrsfs}
\usepackage{amsfonts}
\usepackage{amssymb}
\usepackage{makeidx}
\usepackage{graphicx}
\usepackage{epstopdf}
\usepackage{float}
\usepackage{natbib}
\usepackage{ragged2e}
\usepackage{csquotes}
\usepackage{color}
\usepackage{amsthm}

\journal{Journal of Theoretical Biology}

\newtheorem{lemma}{Lemma}
\newtheorem{theorem}{Theorem}
\newtheorem{corollary}{Corollary}
\newtheorem{proposition}{Proposition}

\newcommand{\PP}{{\mathbb P}}
\newcommand{\MP}{{\mathtt{MP}}}
\newcommand{\AL}{{\mathcal{A}}}

\begin{document}

\begin{frontmatter}

\title{On the accuracy of ancestral sequence reconstruction for ultrametric trees with parsimony}

\author{Lina Herbst\fnref{lh}} 
\author{Mareike Fischer\fnref{mf}}
\address{Institute for Mathematics and Computer Science, Greifswald University, Walther-Rathenau-Straße 47, 17487 Greifswald, Germany}
\fntext[lh]{lina.herbst@uni-greifswald.de}
\fntext[mf]{email@mareikefischer.de}

\begin{abstract}
We examine a mathematical question concerning the reconstruction accuracy of the Fitch algorithm for reconstructing the ancestral
sequence of the most recent common ancestor given a phylogenetic tree
and sequence data for all taxa under consideration. In particular,
for the symmetric 4-state substitution model which is also known as Jukes-Cantor model, we answer
affirmatively a conjecture of Li, Steel and Zhang which states that for any
ultrametric phylogenetic tree and a symmetric model, the Fitch parsimony
method using all terminal taxa is more accurate, or at least as accurate, for
ancestral state reconstruction than using any particular terminal taxon or any particular pair of taxa. This
conjecture had so far only been answered for two-state data by Fischer and
Thatte. Here, we focus on answering the biologically more relevant case with four
states, which corresponds to ancestral sequence reconstruction from DNA or RNA data.
\end{abstract}


\end{frontmatter}

%\linenumbers

\section{Introduction}
\noindent
The reconstruction of ancestral sequences, e.g. DNA-sequences of common ancestors of present-day species, is an important approach in understanding the evolution and origin of these species \citep{moretaxa,liberles,ra2}. There exist various methods to do such reconstructions, e.g the Fitch algorithm \citep{links,phylogenetics,fitch}, which is based on the Maximum Parsimony criterion. However, how reliable is such a reconstruction? \\
Several studies analyzed the reliability, the so-called reconstruction accuracy, of the Fitch algorithm for reconstructing ancestral sequence data of the most recent common ancestor given a phylogenetic tree and sequences for all taxa under consideration \citep{moretaxa,subset,ra1}.
It seems intuitive that the root state is more likely to be conserved for taxa that are closer to the root, since over time more sequence changes can occur. Moreover, one might expect that the reconstruction accuracy is highest when all taxa are taken into account, which was also suggested by earlier simulation studies \citep{estimation}. However, it can be shown that there are cases in which the reconstruction accuracy improves when only a subset of taxa is considered \citep{moretaxa,subset}. In particular, the reconstruction accuracy can even improve when a taxon close to the root is ignored \citep{subset}. \\
Despite these counterintuitive results, in 2008 Li et al. conjectured that for any rooted binary ultrametric phylogenetic tree (i.e. a tree in which all branches have the same distance to the root) and a simple model of evolution, the Fitch algorithm using all taxa for ancestral state reconstruction is at least as accurate as using a single taxon \citep{moretaxa}. Note that ultrametric trees are also often referred to as clocklike trees or molecular clocks. So the conjecture by Li et al. means that under a molecular clock, the reconstruction accuracy is at least as good as the conservation probability of any taxon. Note that under a molecular clock all taxa have the same conversation probability, and that this conjecture provides a lower bound on the reconstruction accuracy for any rooted binary ultrametric phylogenetic tree under a simple model of evolution. Ignoring all data besides the data of one species displays the extreme case of throwing information away. Thus, showing that the conjecture holds is good news for Maximum Parsimony as a criterion for ancestral state reconstruction. \\
In 2009, Fischer and Thatte \citep{subset} proved the conjecture for two-state characters, but it remained unclear if it also holds for 4-state data like DNA or RNA. Thus, the aim of this paper is to consider this biologically relevant case with four states. In particular, we answer the conjecture affirmatively. Additionally, we also prove that the conjecture holds for three-state characters. Along the way, we also prove that the Fitch parsimony method applied to all taxa is always at least as good as applied to any pair of taxa if the underlying tree is clocklike. However, we also show that this does not improve the lower bound induced by single leaves.

\section{Preliminaries}
\noindent
Before we can present our results, we first have to introduce some basic concepts.
Recall that a rooted binary phylogenetic tree on the leaf set $X$ ($|X|=n \geq 2$) is a connected, acyclic graph in which the vertices of degree 1 are called leaves, and in which there is exactly one node $\rho$ of degree 2, which is referred to as root, and all other non-leaf nodes have degree 3. Moreover, in a rooted binary phylogenetic $X$-tree the leaves are bijectively labelled by the elements of $X$. Let each vertex of the tree be assigned a state element of a finite state set $\AL$ with $|\AL|\geq 2$. In particular, we are interested in the biologically relevant case with four states, e.g. $\AL=\{\alpha, \beta, \gamma, \delta\}$, which corresponds for instance to DNA or RNA data.\\ 
The states evolve from $\rho$ by the well-known symmetric $r$-state model $N_r$ with alphabet $\AL=\{\alpha_1, \dots, \alpha_r\}$ \citep{links}. In this model, a state of $\AL$ is selected as the root state with probability $\frac{1}{|\AL|}$. Assume that $e=(u,v)$ is an edge of the tree, and node $u$ is closer to the root than $v$. Then in this model, $p_e$ is the substitution probability on edge $e$: it is the probability that $v$ is in some state $\alpha$ under the condition that $u$ is in a distinct state, say, $\beta$. This is denoted by $\PP(v=\alpha | u=\beta)$. The model is supposed to be symmetric, thus $p_e=\PP(v=\alpha | u=\beta)=\PP(v=\beta | u=\alpha)$. Furthermore, we assume that $0 \leq p_e \leq \frac{1}{|\AL|}$, in particular for four states we have $0 \leq p_e \leq \frac{1}{4}$. The biologically relevant case with four states, namely the $N_4$-model, is also often referred to as Jukes-Cantor-model \citep{jc}. \\
Similar as in \citep{subset,diplom}, we consider ultrametric trees, often known as clocklike trees or molecular clocks by biologists.
It means that the expected number of substitutions from the root to any leaf is the same \citep{phylogenetics}. \\
In this manuscript we reconstruct ancestral states by the Maximum Parsimony criterion with the Fitch algorithm, which we briefly explain now. Assume that we have a rooted binary tree with leaf set $X$. To introduce the Fitch algorithm, we first consider the kind of data we will map onto the leaves of the tree.
The data is given by a character on a leaf set $X$, which is a function $f: X \rightarrow \AL$. Thus, each leaf is assigned a character state. Note that as we consider $X=\{1,\dots,n\}$, we often write $f=f(1)f(2)\dots f(n)$ instead of listing $f(1),\ldots,f(n)$ explicitly. \\
Then the Fitch algorithm \citep{fitch} assigns a set of states to all interior vertices by minimizing the number of changes. The algorithm is based on Fitch's parsimony operation. Therefore, let $\AL$ be a non-empty finite alphabet and let $A,B \subseteq \AL$. Then, Fitch's parsimony operation $*$ is defined by
$$A*B \coloneqq \begin{cases}
A \cap B, & \text{if } A \cap B \neq \emptyset, \\
A \cup B, & \text{otherwise.}
\end{cases}$$
Using this operation, the Fitch algorithm works as follows. Consider all vertices $v$, whose two direct descendants have already been assigned a set, say $A$ and $B$. Then, $v$ is assigned $A*B$. This step is continued upwards along the tree until the root $\rho$ is assigned a set, which is denoted by $\MP(f,T)$. An example can be seen in Figure \ref{fitchexp}.
\begin{figure}[H]
\centering
\includegraphics[scale=0.6]{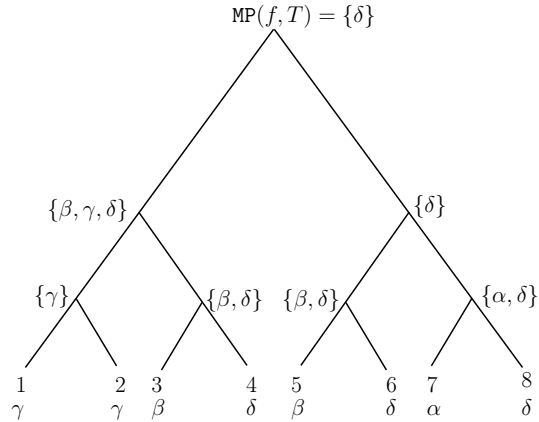}
\caption{Example for the Fitch algorithm for a rooted binary tree and the character $f:\gamma \gamma \beta \delta \beta \delta \alpha \delta$. At first each leaf is assigned the state specified by the character. Then all other vertices whose direct descendants have already been assigned a set are assigned a set by applying the parsimony operation. This step is continued until the root is assigned a set; here $\MP(f,T)=\{\delta\}$.}
\label{fitchexp}
\end{figure}
\noindent
Note that what we call the Fitch algorithm is in fact only one phase of the algorithm, but it is the only part we require to estimate potential root states. For more details we refer to \citep{fitch}. \\
For a 4-state-character there are $2^4-1=15$ possible sets for each interior vertex, since 16 is the cardinality of the power set of an alphabet with four elements minus one for the empty set, i.e.: $\{\alpha\},\{\beta\},\{\gamma\},\{\delta\},\{\alpha,\beta\},\dots,\{\alpha,\beta,\gamma,\delta\}$. \\
We say that the Fitch algorithm unambiguously reconstructs the root state if $|\MP(f,T)|=1$.
Otherwise the root state is reconstructed ambiguously, i.e. the method cannot decide between different states and therefore $|\MP(f,T)| > 1$.
\\
Note that real data usually comes in the form of an alignment, i.e. a sequence of characters, rather than in the form of an individual character. In this case, the Fitch algorithm would consider each character, i.e. each column (\enquote{site}) of the alignment, separately. This is why we focus on the case of a single character and its reconstruction accuracy.

\section{The accuracy of ancestral sequence reconstruction with 4-state characters} \label{4states}
\noindent 
Similar to Li et al., we now define the reconstruction accuracy for all $|\AL|\geq 2$ \citep{moretaxa}.
Therefore, let $\MP(f,T)$ denote the set of character states chosen by the Fitch algorithm as possible root states when applied to character $f$ on tree $T$. \\
Let $\mathcal{R} \subseteq \AL, \alpha \in \mathcal{R}$ and $|\mathcal{R}| \geq 1$. The probability that the root state $\alpha$ evolves on $T$ to a character $f$ for which the Fitch algorithm assigns $\mathcal{R}$ as possible root state set is given by $\PP(\MP(f,T)=\mathcal{R} | \rho=\alpha)$. \\
The reconstruction accuracy is then defined by
\begin{align}RA(X) \coloneqq \sum_{\substack{\mathcal{R} \subseteq \AL \\ \alpha \in \mathcal{R}}} \frac{1}{| \mathcal{R} |} \cdot \PP(\MP(f,T)=\mathcal{R} | \rho=\alpha). \label{RA}
\end{align}
To illustrate this definition, consider the case with $\AL=\{ \alpha, \beta, \gamma, \delta \}$. In this case, the reconstruction accuracy for the Fitch algorithm for ancestral state reconstruction is given by
\begin{align}
RA(X)=&P_{\alpha}(X)+\frac{1}{2} \cdot (P_{\alpha\beta}(X)+P_{\alpha\gamma}(X)+P_{\alpha\delta}(X)) \nonumber \\
&+ \frac{1}{3} \cdot (P_{\alpha\beta\gamma}(X)+P_{\alpha\beta\delta}(X)+P_{\alpha\gamma\delta}(X)) + \frac{1}{4} \cdot P_{\alpha\beta\gamma\delta}(X), \label{RA4_1}
\end{align}
where we define
\begin{align*}
&P_{\alpha}(X) \coloneqq \PP(\MP(f,T)=\{\alpha\}|\rho=\alpha),\\
&P_{\alpha\beta}(X) \coloneqq \PP(\MP(f,T)=\{\alpha,\beta\}|\rho=\alpha),\\
&P_{\alpha\gamma}(X) \coloneqq \PP(\MP(f,T)=\{\alpha,\gamma\}|\rho=\alpha),\\
&P_{\alpha\delta}(X) \coloneqq \PP(\MP(f,T)=\{\alpha,\delta\}|\rho=\alpha),\\
&P_{\alpha\beta\gamma}(X) \coloneqq \PP(\MP(f,T)=\{\alpha,\beta,\gamma\}|\rho=\alpha),\\
&P_{\alpha\beta\delta}(X) \coloneqq \PP(\MP(f,T)=\{\alpha,\beta,\delta\}|\rho=\alpha),\\
&P_{\alpha\gamma\delta}(X) \coloneqq \PP(\MP(f,T)=\{\alpha,\gamma,\delta\}|\rho=\alpha),\\
&P_{\alpha\beta\gamma\delta}(X) \coloneqq \PP(\MP(f,T)=\{\alpha,\beta,\gamma\}|\rho=\alpha).
\end{align*}
The main aim of this manuscript is to show that the reconstruction accuracy for a rooted binary ultrametric phylogenetic tree under the $N_4$-model using all terminal taxa is more accurate, or at least as accurate, for ancestral state reconstruction than using any particular terminal taxon. This provides a lower bound on $RA(X)$, and is stated in the following theorem.
\begin{theorem} \label{theoremRA}
For any rooted binary phylogenetic ultrametric tree and the $N_4$-model, the Fitch algorithm using all terminal taxa is more accurate, or at least as accurate, for ancestral state reconstruction than using any particular terminal taxon, that is
$$RA(X) \geq 1-3p.$$
\end{theorem}
\noindent
The proof of Theorem \ref{theoremRA} requires some more general properties. Therefore, we first turn our attention to the following. If not stated otherwise, we always consider rooted binary ultrametric phylogenetic trees under the $N_4$-model. Due to the symmetry of the model, we can assume without loss of generality that the root is in state $\alpha$, so $\alpha$ evolves along the tree to a character $f$ on $X$. 
Let $p$ be the probability that from the root to one leaf the state changes from $\alpha$ to one specific state in $\AL \setminus \{\alpha\} = \{\beta,\gamma,\delta\}$, i.e. $3p$ is the probability that a given leaf is not in state $\alpha$. \\
Therefore, in the case of the $N_4$-model, $1-3p$ is the probability that the root is in the same state as one leaf, since three different changes ($\alpha \rightarrow \beta,\alpha \rightarrow \gamma, \alpha \rightarrow \delta$) can occur. This is at the same time the reconstruction accuracy when only one leaf is taken into account. The main aim of this paper is to show that $1-3p$ is a lower bound for $RA(X)$; that is considering all taxa under a molecular clock is always better, or as good as, considering just one taxon. \\
As shown in Figure \ref{tree}, every binary tree $T$ can be decomposed into two maximal pending subtrees $T_1$ and $T_2$ with leaf sets $Y_1$ and $Y_2$ ($X=Y_1 \cup Y_2, Y_1 \cap Y_2 = \emptyset$). This is the so-called standard decomposition \citep{phylogenetics}. We denote the children of $\rho$ by $y_1$ and $y_2$, and with probability $p_i$ one specific change occurs from $\rho$ to $y_i$ ($i \in \{1,2\}$). Analogously, one specific change occurs from $y_i$ to any leaf with probability $p_i^{'}$ ($i \in \{1,2\}$). Note that $p$ can then be calculated by all possibilities given for one specific change from $\rho$ to any leaf. Suppose that the root is in state $\alpha$ and leaf $l$ in state $\beta$ (without loss of generality we have $l \in Y_1$). Then there are four different possibilities for a change from $\rho=\alpha$ to $l=\beta$:
\begin{align*}
&\rho=\alpha \rightarrow y_1=\alpha \rightarrow l=\beta, \\
&\rho=\alpha \rightarrow y_1=\beta \rightarrow l=\beta, \\
&\rho=\alpha \rightarrow y_1=\gamma \rightarrow l=\beta, \\
&\rho=\alpha \rightarrow y_1=\delta \rightarrow l=\beta
.
\end{align*}
Thus,
\begin{align}
p &= (1-3p_i)p_i^{'} + p_i(1-3p_i^{'})+p_i p_i^{'}+p_i p_i^{'} \nonumber \\
&=p_i+p_i^{'}-4p_i p_i^{'}. \label{defp}
\end{align}
Furthermore, for $i \in\{1,2\}$ we define $P_i \coloneqq 1-4p_i$, and similarly $P \coloneqq 1-4p$. \\
\begin{figure}[H]
\centering
\includegraphics[scale=0.5]{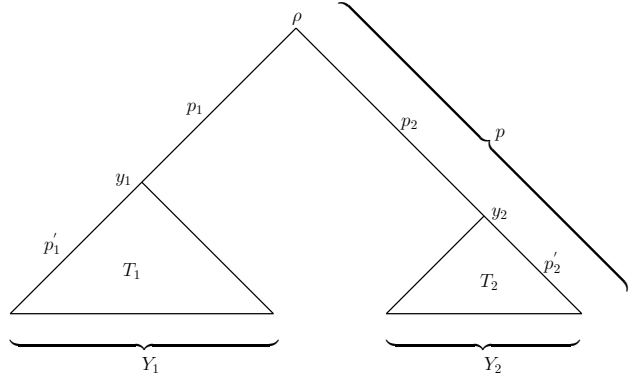}
\caption{Illustration of a rooted binary ultrametric phylogenetic tree and the standard decomposition into its two maximal pending subtrees $T_1$ and $T_2$ \citep{subset}.}
\label{tree}
\end{figure}
\noindent
Under the model assumptions of the $N_4$-model, due to the symmetry, we have that
\begin{align}
&P_{\alpha\beta}(X) = P_{\alpha\gamma}(X) = P_{\alpha\delta}(X), \label{x2} \\
&P_{\alpha\beta\gamma}(X) = P_{\alpha\beta\delta}(X) = P_{\alpha\gamma\delta}(X) \label{x3}, 
\end{align}
since e.g. $$P_{\alpha\beta}(X)=\PP(\MP(f,T)=\{\alpha,\beta\}|\rho=\alpha)=\PP(\MP(f,T)=\{\alpha,\gamma\}|\rho=\alpha)=P_{\alpha\gamma}(X).$$
Therefore by \eqref{RA4_1}, \eqref{x2} and \eqref{x3}, $RA(X)$ can be simplified and becomes
\begin{align}
RA(X)=&P_{\alpha}(X)+\frac{3}{2} P_{\alpha\beta}(X) + P_{\alpha\beta\gamma}(X) + \frac{1}{4} P_{\alpha\beta\gamma\delta}(X). \label{RA4}
\end{align}
Moreover, we define
\begin{align*}
&P_{\beta}(X) \coloneqq \PP(\MP(f,T)=\{\beta\}|\rho=\alpha),\\
&P_{\gamma}(X) \coloneqq \PP(\MP(f,T)=\{\gamma\}|\rho=\alpha),\\
&P_{\delta}(X) \coloneqq \PP(\MP(f,T)=\{\delta\}|\rho=\alpha),\\
&P_{\beta\gamma}(X) \coloneqq \PP(\MP(f,T)=\{\beta,\gamma\}|\rho=\alpha),\\
&P_{\beta\gamma\delta}(X) \coloneqq \PP(\MP(f,T)=\{\beta,\gamma,\delta\}|\rho=\alpha).
\end{align*}
Again, by the symmetry of the $N_4$-model, we obtain
\begin{align}
P_{\beta}(X) = P_{\gamma}(X) = P_{\delta}(X). \label{x1}
\end{align}
Biologically this means that under the assumption that $\alpha$ is the true root state, the probability that $\alpha$ evolves to a character for which the Fitch algorithm assigns $\{\beta\}$ to the root is the same as for $\{\gamma\}$ and $\{\delta\}$, since each specific change occurs with probability $p$. \\
This brings us to our next result, where $P_{\alpha}(X), P_{\beta}(X), P_{\alpha\beta}(X), P_{\beta\gamma}(X), P_{\alpha\beta\gamma}(X)$ and $P_{\beta\gamma\delta}(X)$ are linked to each other.
\begin{lemma} \label{lemmageq}
For any rooted binary phylogenetic tree and the $N_4$-model we have that
\begin{align*}
&P_{\alpha}(X) \geq P_{\beta}(X), \\
&P_{\alpha\beta}(X) \geq P_{\beta\gamma}(X),\\
&P_{\alpha\beta\gamma}(X) \geq P_{\beta\gamma\delta}(X).
\end{align*}
\end{lemma}
\noindent
Note that Lemma \ref{lemmageq} does not require the underlying tree to be ultrametric. \\
The proof of Lemma \ref{lemmageq} is by induction on $n$ and is presented in the appendix. For this proof and also for the proof of Theorem \ref{theoremRA} we state some recursions required for the induction. Therefore, we define $f_{Y_i}$ as a restriction of $f$ to $Y_i \subseteq X$ for $i \in \{1,2\}$: $f_{Y_i} \coloneqq f|_{Y_i}$. For $i \in \{1,2\}$ the probability $P_{(A)}(Y_i)$ to obtain a set $A \in \{ \{\alpha\}, \{\beta\}, \{\alpha,\beta\}, \{\beta,\gamma\}, \{\alpha,\beta,\gamma\}, \{\beta,\gamma,\delta\}, \{\alpha,\beta,\gamma,\delta\} \}$ as estimate state for $y_i$ with the Fitch algorithm under the assumption that $\rho$ is in state $\alpha$ can be defined using the law of total probability:
\begin{align*}
P_{(A)}(Y_i) &\coloneqq \PP(\MP(f_{Y_i},T_i)=A) \\
&= (1-3p_i) \PP(\MP(f_{Y_i},T_i)=A | y_i= \alpha) + p_i \PP(\MP(f_{Y_i},T_i)=A | y_i= \beta) \\
&+ p_i \PP(\MP(f_{Y_i},T_i)=A | y_i= \gamma) + p_i \PP(\MP(f_{Y_i},T_i)=A | y_i= \delta).
\end{align*}
Then with \eqref{x2},\eqref{x3},\eqref{x1} we have:
\begin{align}
&P_{(\alpha)}(Y_i) =(1-3p_i) P_{\alpha}(Y_i) + 3p_i P_{\beta}(Y_i), \label{w1} \\
&P_{(\beta)}(Y_i) = (1-p_i) P_{\beta}(Y_i) + p_i P_{\alpha}(Y_i) = P_{(\gamma)}(Y_i)= P_{(\delta)}(Y_i), \label{w2} \\
&P_{(\alpha\beta)}(Y_i) = (1-2p_i) P_{\alpha\beta}(Y_i) + 2p_i P_{\beta\gamma}(Y_i)=P_{(\alpha\gamma)}(Y_i)=P_{(\alpha\delta)}(Y_i), \label{w3} \\
&P_{(\beta\gamma)}(Y_i) = (1-2p_i) P_{\beta\gamma}(Y_i) + 2p_i P_{\alpha\beta}(Y_i)=P_{(\beta\delta)}(Y_i)=P_{(\gamma\delta)}(Y_i), \label{w4} \\
&P_{(\alpha\beta\gamma)}(Y_i) =(1-p_i) P_{\alpha\beta\gamma}(Y_i) + p_i P_{\beta\gamma\delta}(Y_i)=P_{(\alpha\beta\delta)}(Y_i)=P_{(\alpha\gamma\delta)}(Y_i), \label{w5} \\
&P_{(\beta\gamma\delta)}(Y_i) =(1-3p_i) P_{\beta\gamma\delta}(Y_i) + 3p_i P_{\alpha\beta\gamma}(Y_i), \label{w6} \\
&P_{(\alpha\beta\gamma\delta)}(Y_i) = P_{\alpha\beta\gamma\delta}(Y_i). \label{w7} 
\end{align}
With \eqref{w1}, \eqref{w2}, \eqref{w3}, \eqref{w4}, \eqref{w5}, \eqref{w6} and \eqref{w7} we therefore have
\begin{align}
P_{\alpha}(X) = &P_{(\alpha)}(Y_1) P_{(\alpha)}(Y_2) +3P_{(\alpha)}(Y_1) P_{(\alpha\beta)}(Y_2) + 3P_{(\alpha\beta)}(Y_1) P_{(\alpha)}(Y_2) \nonumber \\
&+3P_{(\alpha)}(Y_1) P_{(\alpha\beta\gamma)}(Y_2) + 3P_{(\alpha\beta\gamma)}(Y_1) P_{(\alpha)}(Y_2) \nonumber \\
&+6P_{(\alpha\beta)}(Y_1) P_{(\alpha\beta)}(Y_2) +3P_{(\alpha\beta)}(Y_1) P_{(\alpha\beta\gamma)}(Y_2) + 3P_{(\alpha\beta\gamma)}(Y_1) P_{(\alpha\beta)}(Y_2) \nonumber \\
&+P_{(\alpha)}(Y_1) P_{(\alpha\beta\gamma\delta)}(Y_2) + P_{(\alpha\beta\gamma\delta)}(Y_1) P_{(\alpha)}(Y_2), \label{alpha} \\
P_{\alpha\beta}(X) = &P_{(\alpha)}(Y_1) P_{(\beta)}(Y_2) + P_{(\beta)}(Y_1) P_{(\alpha)}(Y_2)+P_{(\alpha\beta)}(Y_1) P_{(\alpha\beta)}(Y_2) \nonumber \\
&+2P_{(\alpha\beta)}(Y_1) P_{(\alpha\beta\gamma)}(Y_2) + 2P_{(\alpha\beta\gamma)}(Y_1) P_{(\alpha\beta)}(Y_2) + P_{(\alpha\beta)}(Y_1) P_{(\alpha\beta\gamma\delta)}(Y_2) \nonumber \\
&+ P_{(\alpha\beta\gamma\delta)}(Y_1) P_{(\alpha\beta)}(Y_2) +2P_{(\alpha\beta\gamma)}(Y_1) P_{(\alpha\beta\gamma)}(Y_2), \label{alphabeta} \\
P_{\alpha\beta\gamma}(X) = &P_{(\alpha)}(Y_1) P_{(\beta\gamma)}(Y_2) + P_{(\beta\gamma)}(Y_1) P_{(\alpha)}(Y_2)+ 2P_{(\beta)}(Y_1) P_{(\alpha\beta)}(Y_2) \nonumber \\
&+ 2P_{(\alpha\beta)}(Y_1) P_{(\beta)}(Y_2)+P_{(\alpha\beta\gamma)}(Y_1) P_{(\alpha\beta\gamma)}(Y_2) + P_{(\alpha\beta\gamma)}(Y_1) P_{(\alpha\beta\gamma\delta)}(Y_2) \nonumber \\
&+ P_{(\alpha\beta\gamma\delta)}(Y_1) P_{(\alpha\beta\gamma)}(Y_2) , \label{alphabetagamma} \\
P_{\alpha\beta\gamma\delta}(X) = &P_{(\alpha)}(Y_1) P_{(\beta\gamma\delta)}(Y_2) + P_{(\beta\gamma\delta)}(Y_1) P_{(\alpha)}(Y_2)+ 3P_{(\beta)}(Y_1) P_{(\alpha\beta\gamma)}(Y_2) \nonumber \\
&+ 3P_{(\alpha\beta\gamma)}(Y_1) P_{(\beta)}(Y_2)+ 3P_{(\alpha\beta)}(Y_1) P_{(\beta\gamma)}(Y_2) +3 P_{(\beta\gamma)}(Y_1) P_{(\alpha\beta)}(Y_2) \nonumber \\
&+P_{(\alpha\beta\gamma\delta)}(Y_1) P_{(\alpha\beta\gamma\delta)}(Y_2). \label{alphabetagammadelta}
\end{align}
\noindent
As stated before, all these recursions are needed for the proof of Lemma \ref{lemmageq} and Theorem \ref{theoremRA}. Now, we are in the position to prove Theorem \ref{theoremRA}, our main result, which states a lower bound on $RA(X)$. 
\begin{proof}
The proof is by induction on $n$. In order to show $RA(X) \geq 1-3p$, we define $D(X) \coloneqq RA(X) - (1-3p)$, and show that $D(X)$ is non-negative. 
\\ 
For $n=2$ the subtrees $Y_1$ and $Y_2$ both contain one leaf, and thus
\begin{align*}
D(X)&=P_{\alpha}(X)+\frac{3}{2}  P_{\alpha\beta}(X) + P_{\alpha\beta\gamma}(X) + \frac{1}{4} P_{\alpha\beta\gamma\delta}(X) -1+3p \qquad \text{by } \eqref{RA4} \\
&=P_{\alpha}(X)+\frac{3}{2}  P_{\alpha\beta}(X) -1+3p \\
&\quad \text{since } P_{\alpha\beta\gamma}(X)=P_{\alpha\beta\gamma\delta}(X)=0 \text{ for } n=2 \\
&= (1-3p)^2 + \frac{3}{2}~ 2(1-3p)p -1+3p \\
&=1-6p+9p^2+3p-9p^2-1+3p \\
&=0
.
\end{align*} This shows that $D(X)=RA(X)-(1-3p)=0$ is non-negative and thus $RA(X)=1-3p$, which completes the base case of the induction.
\\
Now, we show by induction that $D(X)$ is non-negative. Suppose that $T$ has $n$ taxa and that $D(X)$ is non-negative for all trees having fewer than $n$ taxa. We define $D_i \coloneqq D(Y_i)=RA(Y_i)-(1-3p_i^{'})$ for $i \in \{1,2\}$. Thus, $D_1$ and $D_2$ are non-negative since $Y_1$ and $Y_2$ contain both fewer than $n$ taxa.
\\
By elementary term conversion we can show that
\begin{align}
8D(X) = &\Bigl( 4P_{(\alpha)}(Y_1) + 10P_{(\alpha\beta)}(Y_1) + 6P_{(\alpha\beta\gamma)}(Y_1) \Bigr) \Bigl( P_{\alpha\beta}(Y_2)-P_{\beta\gamma}(Y_2) \Bigr) P_2 \nonumber \\
&+ \Bigl( 4P_{(\alpha)}(Y_2) + 10P_{(\alpha\beta)}(Y_2) + 6P_{(\alpha\beta\gamma)}(Y_2) \Bigr) \Bigl( P_{\alpha\beta}(Y_1)-P_{\beta\gamma}(Y_1) \Bigr) P_1 \nonumber \\ 
&+ \Bigl( 2P_{(\alpha)}(Y_1) + \frac{16}{3} P_{(\alpha\beta)}(Y_1) + 2P_{(\alpha\beta\gamma)}(Y_1) \Bigr) \Bigl( P_{\alpha\beta\gamma}(Y_2)-P_{\beta\gamma\delta}(Y_2) \Bigr) P_2 \nonumber \\
&+ \Bigl( 2P_{(\alpha)}(Y_2) + \frac{16}{3} P_{(\alpha\beta)}(Y_2) + 2P_{(\alpha\beta\gamma)}(Y_2) \Bigr) \Bigl( P_{\alpha\beta\gamma}(Y_1)-P_{\beta\gamma\delta}(Y_1) \Bigr) P_1 \nonumber \\
&+ 4P_1 D_1 + 4 P_2 D_2 \nonumber \\
&+ \Bigl( \frac{2}{3} P_{(\alpha\beta)}(Y_1) + 2 P_{(\alpha\beta\gamma)}(Y_1) + P_{(\alpha\beta\gamma\delta)}(Y_1) \Bigr) \Bigl(3P + 4 P_2 D_2 \Bigr) \nonumber \\
&+ \Bigl( \frac{2}{3} P_{(\alpha\beta)}(Y_2) + 2 P_{(\alpha\beta\gamma)}(Y_2) + P_{(\alpha\beta\gamma\delta)}(Y_2) \Bigr) \Bigl(3P + 4 P_1 D_1 \Bigr) \label{8D(X)}.
\end{align} The exact conversions can be found in the appendix. \\
Moreover, note that $P_i,P_{(\alpha)}(Y_i),P_{(\alpha\beta)}(Y_i),P_{(\alpha\beta\gamma)}(Y_i),P_{(\alpha\beta\gamma\delta)}(Y_i)$ are all probabilities and therefore are all non-negative for $i \in \{1,2\}$. By Lemma \ref{lemmageq} we have that (for $i \in \{1,2\}$) $P_{\alpha\beta}(Y_i)-P_{\beta\gamma}(Y_i)$ and $P_{\alpha\beta\gamma}(Y_i)-P_{\beta\gamma\delta}(Y_i)$ are non-negative, resulting in \eqref{8D(X)} being non-negative. This implies $D(X) \geq 0$ and thus $RA(X) \geq 1-3p$. This completes the proof. 
\end{proof}
\noindent 
We have shown that the reconstruction accuracy using all terminal taxa is always greater or equal than the conservation probability of one single taxon. Moreover, the base case of the proof of Theorem \ref{theoremRA} provides more insight into the reconstruction accuracy of using 2-taxon trees under the $N_4$-model.
\begin{corollary} \label{twoleaves}
Let $T$ be a rooted binary ultrametric phylogenetic tree on taxon set $X$ with $|X|=2$. Let $p$ denote the probability of change from the root to any leaf under the $N_4$-model. Then,
the reconstruction accuracy for ancestral state reconstruction using the Fitch algorithm is given by
$$RA(X)=1-3p.$$
\end{corollary}
\noindent
Corollary \ref{twoleaves} states the reconstruction accuracy for ancestral state reconstruction with the Fitch algorithm using ultrametric 2-taxon trees, which is the same probability when using one terminal taxon. In the following proposition we show that the reconstruction accuracy with the Fitch algorithm using any two terminal taxa of a taxa set $X$ is also $1-3p$.   
\begin{proposition} \label{twoleavescor}
For any rooted binary phylogenetic ultrametric tree and the $N_4$-model, the reconstruction accuracy for the Fitch algorithm using any two terminal taxa $x_1,x_2 \in X$ for ancestral state reconstruction is given by
$$RA(\{x_1,x_2\}) = 1-3p.$$
\end{proposition}
\begin{proof}
Let $x_1, x_2 \in X$ be two terminal taxa of any rooted binary ultrametric phylogenetic tree $T$. Moreover, we consider the standard decomposition of $T$ into its two maximal pending subtrees $T_1$ and $T_2$ as depicted in Figure \ref{tree}. Thus, the proof is divided  into two cases. \\
In the first case we have without loss of generality $x_1 \in Y_1$ and $x_2 \in Y_2$. By Corollary \ref{twoleaves} the reconstruction accuracy using $x_1$ and $x_2$ is then $RA(\{x_1,x_2\})=1-3p$. 
\\
In the second case we have either $x_1, x_2 \in Y_1$ or $x_1,x_2 \in Y_2$. Thus, without loss of generality we consider $x_1, x_2 \in Y_1$ as depicted in Figure \ref{tree2}. Let $y$ be the last common ancestor of $x_1$ and $x_2$, i.e. the first node that occurs both on the path from $x_1$ to $\rho$ as well as on the path from $x_2$ to $\rho$. Let $\widehat{T}$ be the subtree of $T_1$ that consists of the paths from $y$ to $x_1$ and $x_2$, respectively, as well as all vertices which lie on one of these paths. $\widehat{T}$ is depicted with dotted lines in Figure \ref{tree2}. Thus, the root of $\widehat{T}$ is $y$. In addition, let $\overline{p}$ be the probability for one specific change from $\rho$ to $y$, and let $\widehat{p}$ be the probability for one specific change from $y$ to $x_1$ or $x_2$.
\begin{figure}[H]
\centering
\includegraphics[scale=0.5]{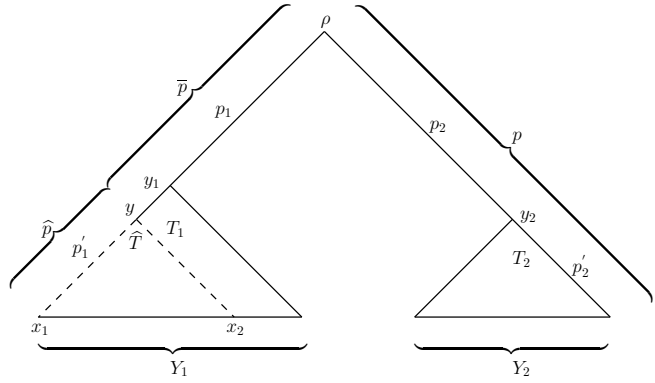}
\caption{Illustration of a rooted binary ultrametric phylogenetic tree and the standard decomposition into its two maximal pending subtrees $T_1$ and $T_2$ \citep{subset}. The subtree $\widehat{T}$ of $T_1$ is represented by the dotted lines, and the root of $\widehat{T}$ is denoted by $y$.}
\label{tree2}
\end{figure}
\noindent
By \eqref{RA4} we have 
\begin{align}
RA(\{x_1,x_2\})=P_{\alpha}(\{x_1,x_2\}) + \frac{3}{2} P_{\alpha\beta}(\{x_1,x_2\}). \label{ra2taxa}
\end{align}
Note that $P_{\alpha\beta\gamma}(\{x_1,x_2\})=P_{\alpha\beta\gamma\delta}(\{x_1,x_2\})=0$ since we cannot obtain sets with more than two elements with the Fitch algorithm when only $x_1$ and $x_2$ are used for the reconstruction.
\\
In the following, we use the notation $f|_{\{x_1,x_2\}}$ for the restriction of character $f$ on taxa $x_1$ and $x_2$. \\
Furthermore, we have
\begin{align}
P_{\alpha}(\{x_1,x_2\}) &= \PP(\MP(f|_{\{x_1,x_2\}},\widehat{T})=\{\alpha\} | \rho=\alpha) \nonumber \\
&= (1-3\overline{p})~ \PP(\MP(f|_{\{x_1,x_2\}},\widehat{T})=\{\alpha\} | y=\alpha,\rho=\alpha) \nonumber \\
&\quad + \overline{p}~ \PP(\MP(f|_{\{x_1,x_2\}},\widehat{T})=\{\alpha\} | y=\beta,\rho=\alpha) \nonumber \\
&\quad + \overline{p}~ \PP(\MP(f|_{\{x_1,x_2\}},\widehat{T})=\{\alpha\} | y=\gamma,\rho=\alpha) \nonumber \\
&\quad + \overline{p}~ \PP(\MP(f|_{\{x_1,x_2\}},\widehat{T})=\{\alpha\} | y=\delta,\rho=\alpha) \nonumber \\
&= (1-3\overline{p})~ \PP(\MP(f|_{\{x_1,x_2\}},\widehat{T})=\{\alpha\} | y=\alpha,\rho=\alpha) \nonumber \\
&\quad + 3\overline{p}~ \PP(\MP(f|_{\{x_1,x_2\}},\widehat{T})=\{\beta\} | y=\alpha,\rho=\alpha) \nonumber \\
&\qquad \text{by the symmetry of the } N_4 \text{-model} \nonumber \\
&= (1-3\overline{p}) (1-3\widehat{p})^2 + 3\overline{p} \widehat{p}^2 \label{2alpha}
\end{align}
Moreover,
\begin{align}
P_{\alpha\beta}(\{x_1,x_2\}) &= \PP(\MP(f|_{\{x_1,x_2\}},\widehat{T})=\{\alpha,\beta\} | \rho=\alpha) \nonumber \\
&= (1-3\overline{p})~ \PP(\MP(f|_{\{x_1,x_2\}},\widehat{T})=\{\alpha,\beta\} | y=\alpha,\rho=\alpha) \nonumber \\
&\quad + \overline{p}~ \PP(\MP(f|_{\{x_1,x_2\}},\widehat{T})=\{\alpha,\beta\} | y=\beta,\rho=\alpha) \nonumber \\
&\quad + \overline{p}~ \PP(\MP(f|_{\{x_1,x_2\}},\widehat{T})=\{\alpha,\beta\} | y=\gamma,\rho=\alpha) \nonumber \\
&\quad + \overline{p}~ \PP(\MP(f|_{\{x_1,x_2\}},\widehat{T})=\{\alpha,\beta\} | y=\delta,\rho=\alpha) \nonumber \\
&= (1-2\overline{p})~ \PP(\MP(f|_{\{x_1,x_2\}},\widehat{T})=\{\alpha,\beta\} | y=\alpha,\rho=\alpha) \nonumber \\
&\quad + 2\overline{p}~ \PP(\MP(f|_{\{x_1,x_2\}},\widehat{T})=\{\beta,\gamma\} | y=\alpha,\rho=\alpha) \nonumber \\
&\qquad \text{by the symmetry of the } N_4 \text{-model} \nonumber \\
&= (1-2\overline{p}) ~2~ (1-3\widehat{p})~ \widehat{p} + 2~\overline{p} ~2~ \widehat{p}^2 \label{2alphabeta}
\end{align}
Thus by \ref{2alpha} and \eqref{2alphabeta}, \eqref{ra2taxa} becomes
\begin{align*}
RA(\{x_1,x_2\})&= (1-3\overline{p}) (1-3\widehat{p})^2 + 3\overline{p} \widehat{p}^2 + \frac{3}{2} \Bigl( (1-2\overline{p}) ~2~ (1-3\widehat{p})~ \widehat{p} + 2~\overline{p} ~2~ \widehat{p}^2 \Bigr) \\
&=1- 3~\overline{p} -3~\widehat{p} + 12~\overline{p}~\widehat{p} \\
&=1-3p \\
&\quad \text{since similar to \eqref{defp} we have that } p=\overline{p}+\widehat{p}-4~\overline{p} \widehat{p}.  
\end{align*}
Therefore, in both cases $RA(\{x_1,x_2\})=1-3p$ which completes the proof. 
\end{proof}
\noindent
This proposition provides us the reconstruction accuracy for the Fitch algorithm when any two terminal taxa are considered. Note that this reconstruction accuracy is the same as when only one terminal taxon is taken into account. Therefore, by Theorem \ref{theoremRA} and Proposition \ref{twoleavescor} we have the following corollary, which states that the lower bound on the reconstruction accuracy holds for any two terminal taxa. In particular, considering two taxa rather than one cannot improve the lower bound given by Theorem \ref{theoremRA}. 
\begin{corollary}\label{cor2}
For any rooted binary phylogenetic ultrametric tree and the $N_4$-model, the Fitch algorithm using all terminal taxa is more accurate, or at least as accurate, for ancestral state reconstruction than using any two terminal taxa, that is
$$RA(X) \geq 1-3p.$$
\end{corollary}
\noindent
This statement completes Section \ref{4states}, and we now have a look on similar results obtained for the $N_3$-model.
\\ \\
\section{The accuracy of ancestral sequence reconstruction with 3-state characters}
\noindent 
Under the same assumptions as for the 4-state model, similar results can be obtained for the 3-state alphabet $\AL=\{\alpha,\beta,\gamma\}$. In this case, the reconstruction accuracy is given by
\begin{align}
RA(X)=&P_{\alpha}(X)+\frac{1}{2} \cdot (P_{\alpha\beta}(X)+P_{\alpha\gamma}(X))+ \frac{1}{3} \cdot P_{\alpha\beta\gamma}(X) \nonumber \\
=&P_{\alpha}(X)+P_{\alpha\beta}(X)+ \frac{1}{3} \cdot P_{\alpha\beta\gamma}(X). \label{ra3states}
\end{align} Then Theorem \ref{theoremRA3} and Lemma \ref{lemmageq3} can be formulated similarly to the statements before. Both proofs are left out, since they can be done analogously. However, we want to emphasize that the conjecture stated by Li et al. also holds for the $N_3$-model.
\begin{theorem} \label{theoremRA3}
For any rooted binary phylogenetic ultrametric tree and the $N_3$-model, the Fitch algorithm using all terminal taxa is more accurate, or at least as accurate, for ancestral state reconstruction than using any particular terminal taxon, that is
$$RA(X) \geq 1-2p.$$
\end{theorem}
\noindent
By Theorem \ref{theoremRA3}, a lower bound on $RA(X)$ for rooted binary ultrametric phylogenetic trees is also given for $\AL=\{\alpha,\beta,\gamma\}$. 
\\
Note that the analogs of Lemma \ref{lemmageq}, Corollary \ref{twoleaves}, Proposition \ref{twoleavescor} and Corollary \ref{cor2} also hold under the $N_3$-model. In particular, the reconstruction accuracy for ultrametric trees is then at least $1-2p$. The exact statements and their proofs can be found in the appendix.

\section{Conclusion and Discussion}
\noindent
In this paper we considered the reconstruction accuracy of the Fitch algorithm for ancestral state reconstruction. In particular, we analyzed rooted binary ultrametric phylogenetic trees under the $N_4$-model. For an ultrametric tree the probability of a change from the root to any leaf is the same. For such trees, we investigated a lower bound on the reconstruction accuracy by answering affirmatively the conjecture by Li, Steel and Zhang, which stated that for rooted binary ultrametric phylogenetic trees under the symmetric $N_r$-model the reconstruction accuracy using all terminal taxa is at least as high as the conservation probability of any leaf. In 2009, Fischer and Thatte had already shown that this conjecture holds for two-state characters, but it remained unknown whether this result could be extended to three or more character states. In particular, the biologically relevant case of $r=4$, which corresponds to the DNA- or RNA-alphabet, remained unclear.
\\
The main result of this manuscript is the proof of the conjecture for $r=4$, which provides a lower bound on the reconstruction accuracy. As mentioned before, the conjecture also holds for the $N_3$-model. In the past, several studies showed that in some cases, the Fitch algorithm provides better results when some data are disregarded \citep{moretaxa,subset}. This led to a critical view on Maximum Parsimony as a method for ancestral state reconstruction. But as we have shown here, at least for ultrametric trees, the extreme case of disregarding all data except for one or two leaves can never improve the reconstruction accuracy of the Fitch algorithm. In this sense, our results are good news for Maximum Parsimony as a method for ancestral state reconstruction. \\ 
To conclude, the generalization to the $N_r$-model for $r > 4$ is still open, but we conjecture that it also holds.

\section{Appendix}
\noindent
\textbf{Proof of Lemma \ref{lemmageq}}
To prove Lemma \ref{lemmageq} we show that for any rooted binary phylogenetic tree $T$ under a symmetric 4-state substitution model
\begin{align}
&P_{\alpha}(X) \geq P_{\beta}(X), \label{lemma1} \\
&P_{\alpha\beta}(X) \geq P_{\beta\gamma}(X), \label{lemma2} \\
&P_{\alpha\beta\gamma}(X) \geq P_{\beta\gamma\delta}(X)  \label{lemma3}
\end{align} by induction on $n$. For $n=2$ the subtrees $Y_1$ and $Y_2$ both contain one leaf, and hence $p=p_1=p_2$ leads to
\begin{align*}
&P_{\alpha}(X)= (1-3p)^2, \\
&P_{\beta}(X)= p^2, \\
&P_{\alpha\beta}(X)= 2(1-3p)p, \\
&P_{\beta\gamma}(X)= 2p^2, \\
&P_{\alpha\beta\gamma}(X)= 0, \\
&P_{\beta\gamma\delta}(X)= 0.
\end{align*} Therefore
\begin{align*}
P_{\alpha}(X) - P_{\beta}(X) &= (1-3p)^2 - p^2 = 1-6p+9p^2-p^2= 1-6p+8p^2 \\
&=\underbrace{(1-4p)}_{\geq 0} \underbrace{(1-2p)}_{\geq 0} \geq 0 \text{ as } p \leq \frac{1}{4}.
\end{align*}
Moreover 
\begin{align*}
P_{\alpha\beta}(X) - P_{\beta\gamma}(X) &= 2(1-3p)p - 2p^2 = 2p \underbrace{(1-4p)}_{\geq 0} \geq 0 \text{ as } p \leq \frac{1}{4},
\end{align*} and
\begin{align*}
P_{\alpha\beta\gamma}(X) - P_{\beta\gamma\delta}(X) = 0-0=0 \geq 0,
\end{align*} which completes the base case of the induction. For the inductive step we first state some more recursions using \eqref{w2}, \eqref{w3}, \eqref{w4}, \eqref{w5} and \eqref{w6}:
\begin{align}
P_{\beta}(X)= &P_{(\beta)}(Y_1) P_{(\beta)}(Y_2) + P_{(\beta)}(Y_1) P_{(\alpha\beta)}(Y_2) + P_{(\alpha\beta)}(Y_1) P_{(\beta)}(Y_2) \nonumber \\
&+2P_{(\beta)}(Y_1) P_{(\beta\gamma)}(Y_2) + 2P_{(\beta\gamma)}(Y_1) P_{(\beta)}(Y_2) + 2P_{(\alpha\beta)}(Y_1) P_{(\beta\gamma)}(Y_2) \nonumber \\
&+ 2P_{(\beta\gamma)}(Y_1) P_{(\alpha\beta)}(Y_2) + 2P_{(\beta\gamma)}(Y_1) P_{(\beta\gamma)}(Y_2) +2P_{(\beta)}(Y_1) P_{(\alpha\beta\gamma)}(Y_2) \nonumber \\
&+ 2P_{(\alpha\beta\gamma)}(Y_1) P_{(\beta)}(Y_2) + P_{(\beta)}(Y_1) P_{(\beta\gamma\delta)}(Y_2) + P_{(\beta\gamma\delta)}(Y_1) P_{(\beta)}(Y_2) \nonumber \\
&+ 2P_{(\beta\gamma)}(Y_1) P_{(\alpha\beta\gamma)}(Y_2) + 2P_{(\alpha\beta\gamma)}(Y_1) P_{(\beta\gamma)}(Y_2) + P_{(\alpha\beta)}(Y_1) P_{(\beta\gamma\delta)}(Y_2) \nonumber \\
&+ P_{(\beta\gamma\delta)}(Y_1) P_{(\alpha\beta)}(Y_2) + P_{(\beta)}(Y_1) P_{(\alpha\beta\gamma\delta)}(Y_2) + P_{(\alpha\beta\gamma\delta)}(Y_1) P_{(\beta)}(Y_2), \label{Beta}
\end{align}
\begin{align}
P_{\beta\gamma}(X) = &P_{(\beta)}(Y_1) P_{(\gamma)}(Y_2) + P_{(\gamma)}(Y_1) P_{(\beta)}(Y_2)+ P_{(\beta\gamma)}(Y_1) P_{(\beta\gamma)}(Y_2) \nonumber \\
&+P_{(\beta\gamma)}(Y_1) P_{(\alpha\beta\gamma)}(Y_2) + P_{(\alpha\beta\gamma)}(Y_1) P_{(\beta\gamma)}(Y_2) + P_{(\beta\gamma)}(Y_1) P_{(\beta\gamma\delta)}(Y_2)  \nonumber \\
&+ P_{(\beta\gamma\delta)}(Y_1) P_{(\beta\gamma)}(Y_2) + P_{(\alpha\beta\gamma)}(Y_1) P_{(\beta\gamma\delta)}(Y_2) + P_{(\beta\gamma\delta)}(Y_1) P_{(\alpha\beta\gamma)}(Y_2) \nonumber \\
&+P_{(\beta\gamma)}(Y_1) P_{(\alpha\beta\gamma\delta)}(Y_2) + P_{(\alpha\beta\gamma\delta)}(Y_1) P_{(\beta\gamma)}(Y_2), \label{betagamma}
\end{align}
\begin{align}
P_{\beta\gamma\delta}(X) = &3P_{(\beta)}(Y_1) P_{(\beta\gamma)}(Y_2) + 3P_{(\beta\gamma)}(Y_1) P_{(\beta)}(Y_2)+P_{(\beta\gamma\delta)}(Y_1) P_{(\beta\gamma\delta)}(Y_2) \nonumber \\
&+P_{(\beta\gamma\delta)}(Y_1) P_{(\alpha\beta\gamma\delta)}(Y_2) + P_{(\alpha\beta\gamma\delta)}(Y_1) P_{(\beta\gamma\delta)}(Y_2) \label{betagammadelta}
.
\end{align} Moreover we have that for $i \in \{1,2\}$
\begin{align}
P_{(\alpha)}(Y_i)-P_{(\beta)}(Y_i)
&= (1-3p_i) P_{\alpha}(Y_i) + 3p_i P_{\beta}(Y_i) - (1-p_i) P_{\beta}(Y_i) - p_i P_{\alpha}(Y_i) \nonumber \\
&\qquad \text{by } \eqref{w1}, \eqref{w2} \nonumber \\
&=(1-4p_i) P_{\alpha}(Y_i) - (1-4p_i) P_{\beta}(Y_i) \nonumber \\
&=P_i \Bigl( P_{\alpha}(Y_i)-P_{\beta}(Y_i) \Bigr) \label{4(a)(b)} \\
&\qquad \text{by the definition of } P_i \nonumber
\end{align} and thus
\begin{align}
P_{(\alpha)}(Y_i) = P_{(\beta)}(Y_i) + P_i \Bigl( P_{\alpha}(Y_i)-P_{\beta}(Y_i) \Bigr). \label{4(a)}
\end{align}
In the same manner by \eqref{w3}, \eqref{w4}, \eqref{w5} and \eqref{w6} we can see that
\begin{align}
&P_{(\alpha\beta)}(Y_i)-P_{(\beta\gamma)}(Y_i)= P_i \Bigl( P_{\alpha\beta}(Y_i)-P_{\beta\gamma}(Y_i) \Bigr), \label{4(ab)(bc)}\\
&P_{(\alpha\beta\gamma)}(Y_i)-P_{(\beta\gamma\delta)}(Y_i)=P_i \Bigl( P_{\alpha\beta\gamma}(Y_i)-P_{\beta\gamma\delta}(Y_i) \Bigr). \label{4(abc)(bcd)}
\end{align} Therefore
\begin{align}
&P_{(\alpha\beta)}(Y_i) = P_{(\beta\gamma)}(Y_i) + P_i \Bigl( P_{\alpha\beta}(Y_i)-P_{\beta\gamma}(Y_i) \Bigr), \label{4(ab)} \\
&P_{(\alpha\beta\gamma)}(Y_i) = P_{(\beta\gamma\delta)}(Y_i) + P_i \Bigl( P_{\alpha\beta\gamma}(Y_i)-P_{\beta\gamma\delta}(Y_i) \Bigr). \label{4(abc)}
\end{align} Additionally we have the following: choose sets $A_1, A_2$ from $\{\{\alpha\},\{\alpha\beta\},\{\alpha\beta\gamma\}\}$ and $B_1, B_2$ from $\{\{\beta\},\{\beta\gamma\},\{\beta\gamma\delta\}\}$ such that for $i \in \{1,2\}$ $|A_i|=|B_i|$, respectively. Then we have that
\begin{align}
&P_{(A_1)}(Y_1)P_{(A_2)}(Y_2)-P_{(B_1)}(Y_1)P_{(B_2)}(Y_2) \nonumber \nonumber \\
= &\Bigl( P_{(B_1)}(Y_1) + P_1 \Bigl( P_{A_1}(Y_1)-P_{B_1}(Y_1) \Bigr) \Bigr)
\Bigl( P_{(B_2)}(Y_2) + P_2 \Bigl( P_{A_2}(Y_2)-P_{B_2}(Y_2) \Bigr) \Bigr) \nonumber \\
&- P_{(B_1)}(Y_1)P_{(B_2)}(Y_2) \nonumber \\
&\text{by } \eqref{4(a)}, \eqref{4(ab)} \text{ or } \eqref{4(abc)} \nonumber \\
= &P_{(B_1)}(Y_1)P_{(B_2)}(Y_2) + P_{(B_1)}(Y_1)P_2 \Bigl( P_{A_2}(Y_2)-P_{B_2}(Y_2) \Bigr) + P_{(B_2)}(Y_2) P_1 \Bigl( P_{A_1}(Y_1)-P_{B_1}(Y_1) \Bigr) \nonumber \\
&+ P_1P_2 \Bigl( P_{A_1}(Y_1)-P_{B_1}(Y_1) \Bigr) \Bigl( P_{A_2}(Y_2)-P_{B_2}(Y_2) \Bigr) - P_{(B_1)}(Y_1)P_{(B_2)}(Y_2) \nonumber \\
= &P_{(B_1)}(Y_1)P_2 \Bigl( P_{A_2}(Y_2)-P_{B_2}(Y_2) \Bigr) + P_{(B_2)}(Y_2) P_1 \Bigl( P_{A_1}(Y_1)-P_{B_1}(Y_1) \Bigr) \nonumber \\
&+ P_1P_2 \Bigl( P_{A_1}(Y_1)-P_{B_1}(Y_1) \Bigr) \Bigl( P_{A_2}(Y_2)-P_{B_2}(Y_2) \Bigr) . \label{4(a)(a)(b)(b)}
\end{align} Now suppose that $T$ has $n$ taxa and that \eqref{lemma1},\eqref{lemma2} and \eqref{lemma3} are true for all trees having fewer that $n$ taxa. Note that therefore \eqref{4(a)(a)(b)(b)} is non-negative, since $Y_1$ and $Y_2$ contain both fewer than than $n$ taxa. Then
\begin{align*}
&P_{\alpha}(X) - P_{\beta}(X) \\ 
= &P_{(\alpha)}(Y_1) P_{(\alpha)}(Y_2) + 3P_{(\alpha)}(Y_1) P_{(\alpha\beta)}(Y_2) + 3P_{(\alpha\beta)}(Y_1) P_{(\alpha)}(Y_2) \\
&+ 3P_{(\alpha)}(Y_1) P_{(\alpha\beta\gamma)}(Y_2) + 3P_{(\alpha\beta\gamma)}(Y_1) P_{(\alpha)}(Y_2) + 6P_{(\alpha\beta)}(Y_1) P_{(\alpha\beta)}(Y_2) \\
&+3P_{(\alpha\beta)}(Y_1) P_{(\alpha\beta\gamma)}(Y_2) + 3P_{(\alpha\beta\gamma)}(Y_1) P_{(\alpha\beta)}(Y_2) + P_{(\alpha)}(Y_1) P_{(\alpha\beta\gamma\delta)}(Y_2)\\
&+ P_{(\alpha\beta\gamma\delta)}(Y_1) P_{(\alpha)}(Y_2) - P_{(\beta)}(Y_1) P_{(\beta)}(Y_2) - P_{(\beta)}(Y_1) P_{(\alpha\beta)}(Y_2) \\
&- P_{(\alpha\beta)}(Y_1) P_{(\beta)}(Y_2) - 2P_{(\beta)}(Y_1) P_{(\beta\gamma)}(Y_2) - 2P_{(\beta\gamma)}(Y_1) P_{(\beta)}(Y_2) \\
& - 2P_{(\alpha\beta)}(Y_1) P_{(\beta\gamma)}(Y_2) - 2P_{(\beta\gamma)}(Y_1) P_{(\alpha\beta)}(Y_2) - 2P_{(\beta\gamma)}(Y_1) P_{(\beta\gamma)}(Y_2) \\
&-2P_{(\beta)}(Y_1) P_{(\alpha\beta\gamma)}(Y_2) - 2P_{(\alpha\beta\gamma)}(Y_1) P_{(\beta)}(Y_2) - P_{(\beta)}(Y_1) P_{(\beta\gamma\delta)}(Y_2) \\
&-P_{(\beta\gamma\delta)}(Y_1) P_{(\beta)}(Y_2) - 2P_{(\beta\gamma)}(Y_1) P_{(\alpha\beta\gamma)}(Y_2) - 2P_{(\alpha\beta\gamma)}(Y_1) P_{(\beta\gamma)}(Y_2) \\
&-P_{(\alpha\beta)}(Y_1) P_{(\beta\gamma\delta)}(Y_2) - P_{(\beta\gamma\delta)}(Y_1) P_{(\alpha\beta)}(Y_2) -P_{(\beta)}(Y_1) P_{(\alpha\beta\gamma\delta)}(Y_2) \\
&- P_{(\alpha\beta\gamma\delta)}(Y_1) P_{(\beta)}(Y_2) \\
&\text{by } \eqref{alpha} \text{ and } \eqref{Beta} 
\end{align*}
\begin{align*}
=& 
P_{(\alpha)}(Y_1) P_{(\alpha)}(Y_2) - P_{(\beta)}(Y_1) P_{(\beta)}(Y_2)+2P_{(\alpha\beta)}(Y_1) P_{(\alpha\beta)}(Y_2) - 2P_{(\beta\gamma)}(Y_1) P_{(\beta\gamma)}(Y_2) \\
&+ P_{(\alpha\beta)}(Y_2) \Bigl( P_{(\alpha)}(Y_1) - P_{(\beta)}(Y_1) \Bigr) + P_{(\alpha\beta)}(Y_1) \Bigl( P_{(\alpha)}(Y_2) - P_{(\beta)}(Y_2) \Bigr) \\
&+ 2P_{(\alpha)}(Y_1) P_{(\alpha\beta)}(Y_2) - 2P_{(\beta)}(Y_1) P_{(\beta\gamma)}(Y_2) + 2P_{(\alpha\beta)}(Y_1) P_{(\alpha)}(Y_2) - 2P_{(\beta\gamma)}(Y_1) P_{(\beta)}(Y_2) \\
&+ 2P_{(\alpha\beta\gamma)}(Y_2) \Bigl( P_{(\alpha)}(Y_1) - P_{(\beta)}(Y_1) \Bigr) + 2P_{(\alpha\beta\gamma)}(Y_1) \Bigl( P_{(\alpha)}(Y_2) - P_{(\beta)}(Y_2) \Bigr) \\
&+ P_{(\alpha)}(Y_1) P_{(\alpha\beta\gamma)}(Y_2) - P_{(\beta)}(Y_1) P_{(\beta\gamma\delta)}(Y_2) + P_{(\alpha\beta\gamma)}(Y_1) P_{(\alpha)}(Y_2) - P_{(\beta\gamma\delta)}(Y_1) P_{(\beta)}(Y_2) \\
&+ 2P_{(\alpha\beta)}(Y_2) \Bigl( P_{(\alpha\beta)}(Y_1) - P_{(\beta\gamma)}(Y_1) \Bigr) + 2P_{(\alpha\beta)}(Y_1) \Bigl( P_{(\alpha\beta)}(Y_2) - P_{(\beta\gamma)}(Y_2) \Bigr) \\
&+ 2P_{(\alpha\beta\gamma)}(Y_2) \Bigl( P_{(\alpha\beta)}(Y_1) - P_{(\beta\gamma)}(Y_1) \Bigr) + 2P_{(\alpha\beta\gamma)}(Y_1) \Bigl( P_{(\alpha\beta)}(Y_2) - P_{(\beta\gamma)}(Y_2) \Bigr) \\
&+ P_{(\alpha\beta)}(Y_2) \Bigl( P_{(\alpha\beta\gamma)}(Y_1) - P_{(\beta\gamma\delta)}(Y_1) \Bigr) + P_{(\alpha\beta)}(Y_1) \Bigl( P_{(\alpha\beta\gamma)}(Y_2) - P_{(\beta\gamma\delta)}(Y_2) \Bigr) \\
&+ P_{(\alpha\beta\gamma\delta)}(Y_2) \Bigl( P_{(\alpha)}(Y_1) - P_{(\beta)}(Y_1) \Bigr) + P_{(\alpha\beta\gamma\delta)}(Y_1) \Bigl( P_{(\alpha)}(Y_2) - P_{(\beta)}(Y_2) \Bigr)
\end{align*}
\begin{align*}
=&
P_{(\alpha)}(Y_1) P_{(\alpha)}(Y_2) - P_{(\beta)}(Y_1) P_{(\beta)}(Y_2) +2 \Bigl( P_{(\alpha\beta)}(Y_1) P_{(\alpha\beta)}(Y_2) - P_{(\beta\gamma)}(Y_1) P_{(\beta\gamma)}(Y_2) \Bigr) \\
&+ 2 \Bigl( P_{(\alpha)}(Y_1) P_{(\alpha\beta)}(Y_2) - P_{(\beta)}(Y_1) P_{(\beta\gamma)}(Y_2) \Bigr) \\
&+ 2 \Bigl( P_{(\alpha\beta)}(Y_1) P_{(\alpha)}(Y_2) - 2P_{(\beta\gamma)}(Y_1) P_{(\beta)}(Y_2) \Bigr) \\
&+ P_{(\alpha)}(Y_1) P_{(\alpha\beta\gamma)}(Y_2) - P_{(\beta)}(Y_1) P_{(\beta\gamma\delta)}(Y_2) \\
&+ P_{(\alpha\beta\gamma)}(Y_1) P_{(\alpha)}(Y_2) - P_{(\beta\gamma\delta)}(Y_1) P_{(\beta)}(Y_2) \\
&+ \Bigl( P_{(\alpha)}(Y_1) - P_{(\beta)}(Y_1) \Bigr) \Bigl( P_{(\alpha\beta)}(Y_2) +2P_{(\alpha\beta\gamma)}(Y_2) + P_{(\alpha\beta\gamma\delta)}(Y_2) \Bigr) \\
&+ \Bigl( P_{(\alpha)}(Y_2) - P_{(\beta)}(Y_2) \Bigr) \Bigl( P_{(\alpha\beta)}(Y_1) +2P_{(\alpha\beta\gamma)}(Y_1) + P_{(\alpha\beta\gamma\delta)}(Y_1) \Bigr) \\
&+ \Bigl( P_{(\alpha\beta)}(Y_1) - P_{(\beta\gamma)}(Y_1) \Bigr) \Bigl( 2P_{(\alpha\beta)}(Y_2) +2P_{(\alpha\beta\gamma)}(Y_2) \Bigr) \\
&+ \Bigl( P_{(\alpha\beta)}(Y_2) - P_{(\beta\gamma)}(Y_2) \Bigr) \Bigl( 2P_{(\alpha\beta)}(Y_1) +2P_{(\alpha\beta\gamma)}(Y_1) \Bigr) \\
&+ \Bigl( P_{(\alpha\beta\gamma)}(Y_1) - P_{(\beta\gamma\delta)}(Y_1) \Bigr) P_{(\alpha\beta)}(Y_2)+ \Bigl( P_{(\alpha\beta\gamma)}(Y_2) - P_{(\beta\gamma\delta)}(Y_2) \Bigr) P_{(\alpha\beta)}(Y_1)
\end{align*}
\begin{align*}
=&
P_{(\alpha)}(Y_1) P_{(\alpha)}(Y_2) - P_{(\beta)}(Y_1) P_{(\beta)}(Y_2)+ 2 \Bigl( P_{(\alpha\beta)}(Y_1) P_{(\alpha\beta)}(Y_2) - P_{(\beta\gamma)}(Y_1) P_{(\beta\gamma)}(Y_2) \Bigr) \\
&+ 2 \Bigl( P_{(\alpha)}(Y_1) P_{(\alpha\beta)}(Y_2) - P_{(\beta)}(Y_1) P_{(\beta\gamma)}(Y_2) \Bigr) \\
&+ 2 \Bigl( P_{(\alpha\beta)}(Y_1) P_{(\alpha)}(Y_2) - P_{(\beta\gamma)}(Y_1) P_{(\beta)}(Y_2) \Bigr) \\
&+ P_{(\alpha)}(Y_1) P_{(\alpha\beta\gamma)}(Y_2) - P_{(\beta)}(Y_1) P_{(\beta\gamma\delta)}(Y_2) \\
&+ P_{(\alpha\beta\gamma)}(Y_1) P_{(\alpha)}(Y_2) - P_{(\beta\gamma\delta)}(Y_1) P_{(\beta)}(Y_2) \\
&+ P_1 \Bigl( P_{\alpha}(Y_1) - P_{\beta}(Y_1) \Bigr) \Bigl( P_{(\alpha\beta)}(Y_2) +2P_{(\alpha\beta\gamma)}(Y_2) + P_{(\alpha\beta\gamma\delta)}(Y_2) \Bigr) \\
&+ P_2 \Bigl( P_{\alpha}(Y_2) - P_{\beta}(Y_2) \Bigr) \Bigl( P_{(\alpha\beta)}(Y_1) +2P_{(\alpha\beta\gamma)}(Y_1) + P_{(\alpha\beta\gamma\delta)}(Y_1) \Bigr) \\
&+ P_1 \Bigl( P_{\alpha\beta}(Y_1) - P_{\beta\gamma}(Y_1) \Bigr) \Bigl( 2P_{(\alpha\beta)}(Y_2) +2P_{(\alpha\beta\gamma)}(Y_2) \Bigr) \\
&+ P_2 \Bigl( P_{\alpha\beta}(Y_2) - P_{\beta\gamma}(Y_2) \Bigr) \Bigl( 2P_{(\alpha\beta)}(Y_1) +2P_{(\alpha\beta\gamma)}(Y_1) \Bigr) \\
&+ P_1 \Bigl( P_{\alpha\beta\gamma}(Y_1) - P_{\beta\gamma\delta}(Y_1) \Bigr) P_{(\alpha\beta)}(Y_2) + P_2  \Bigl( P_{\alpha\beta\delta}(Y_2) - P_{\beta\gamma\delta}(Y_2) \Bigr) P_{(\alpha\beta)}(Y_1) \\
&\text{by } \eqref{4(a)(b)}, \eqref{4(ab)(bc)} \text{ and } \eqref{4(abc)(bcd)} 
.
\end{align*} By \eqref{4(a)(a)(b)(b)} and the inductive assumption this term is non-negative, and therefore concludes the proof for $P_{\alpha}(X) \geq P_{\beta}(X)$. We now proceed with the second part of Lemma \ref{lemmageq}.
\begin{align*}
&P_{\alpha\beta}(X) - P_{\beta\gamma}(X) \\
=&P_{(\alpha)}(Y_1) P_{(\beta)}(Y_2) + P_{(\beta)}(Y_1) P_{(\alpha)}(Y_2)+ P_{(\alpha\beta)}(Y_1) P_{(\alpha\beta)}(Y_2) \\
&+ 2P_{(\alpha\beta)}(Y_1) P_{(\alpha\beta\gamma)}(Y_2) + 2P_{(\alpha\beta\gamma)}(Y_1) P_{(\alpha\beta)}(Y_2) + P_{(\alpha\beta)}(Y_1) P_{(\alpha\beta\gamma\delta)}(Y_2) \\
&+P_{(\alpha\beta\gamma\delta)}(Y_1) P_{(\alpha\beta)}(Y_2) + 2P_{(\alpha\beta\gamma)}(Y_1) P_{(\alpha\beta\gamma)}(Y_2) - 2P_{(\beta)}(Y_1) P_{(\beta)}(Y_2) \\
&- P_{(\beta\gamma)}(Y_1) P_{(\beta\gamma)}(Y_2) - P_{(\beta\gamma)}(Y_1) P_{(\alpha\beta\gamma)}(Y_2) - P_{(\alpha\beta\gamma)}(Y_1) P_{(\beta\gamma)}(Y_2) \\
&- P_{(\beta\gamma)}(Y_1) P_{(\beta\gamma\delta)}(Y_2) - P_{(\beta\gamma\delta)}(Y_1) P_{(\beta\gamma)}(Y_2) - P_{(\alpha\beta\gamma)}(Y_1) P_{(\beta\gamma\delta)}(Y_2) \\
&- P_{(\beta\gamma\delta)}(Y_1) P_{(\alpha\beta\gamma)}(Y_2) - P_{(\beta\gamma)}(Y_1) P_{(\alpha\beta\gamma\delta)}(Y_2) - P_{(\alpha\beta\gamma\delta)}(Y_1) P_{(\beta\gamma)}(Y_2) \\
&\text{by } \eqref{w2}, \eqref{alphabeta} \text{ and } \eqref{betagamma} 
\end{align*}
\begin{align*}
= &P_{(\beta)}(Y_2) \Bigl( P_{(\alpha)}(Y_1) - P_{(\beta)}(Y_1) \Bigr) + P_{(\beta)}(Y_1) \Bigl( P_{(\alpha)}(Y_2) - P_{(\beta)}(Y_2) \Bigr) \\
&+ P_{(\alpha\beta)}(Y_1) P_{(\alpha\beta)}(Y_2) - P_{(\beta\gamma)}(Y_1) P_{(\beta\gamma)}(Y_2) \\
&+ P_{(\alpha\beta\gamma)}(Y_2) \Bigl( P_{(\alpha\beta)}(Y_1) - P_{(\beta\gamma)}(Y_1) \Bigr) + P_{(\alpha\beta\gamma)}(Y_1) \Bigl( P_{(\alpha\beta)}(Y_2) - P_{(\beta\gamma)}(Y_2) \Bigr) \\
&+ P_{(\alpha\beta)}(Y_1) P_{(\alpha\beta\gamma)}(Y_2) - P_{(\beta\gamma)}(Y_1) P_{(\beta\gamma\delta)}(Y_2) \\
&+ P_{(\alpha\beta\gamma)}(Y_1) P_{(\alpha\beta)}(Y_2) - P_{(\beta\gamma\delta)}(Y_1) P_{(\beta\gamma)}(Y_2) \\
&+ P_{(\alpha\beta\gamma)}(Y_2) \Bigl( P_{(\alpha\beta\gamma)}(Y_1) - P_{(\beta\gamma\delta)}(Y_1) \Bigr) + P_{(\alpha\beta\gamma)}(Y_1) \Bigl( P_{(\alpha\beta\gamma)}(Y_2) - P_{(\beta\gamma\delta)}(Y_2) \Bigr) \\
&+ P_{(\alpha\beta\gamma\delta)}(Y_2) \Bigl( P_{(\alpha\beta)}(Y_1) - P_{(\beta\gamma)}(Y_1) \Bigr) + P_{(\alpha\beta\gamma\delta)}(Y_1) \Bigl( P_{(\alpha\beta)}(Y_2) - P_{(\beta\gamma)}(Y_2) \Bigr)
\end{align*}
\begin{align*}
=& P_{(\alpha\beta)}(Y_1) P_{(\alpha\beta)}(Y_2) - P_{(\beta\gamma)}(Y_1) P_{(\beta\gamma)}(Y_2) \\
&+ P_{(\alpha\beta)}(Y_1) P_{(\alpha\beta\gamma)}(Y_2) - P_{(\beta\gamma)}(Y_1) P_{(\beta\gamma\delta)}(Y_2) \\
&+ P_{(\alpha\beta\gamma)}(Y_1) P_{(\alpha\beta)}(Y_2) - P_{(\beta\gamma\delta)}(Y_1) P_{(\beta\gamma)}(Y_2) \\
&+P_{(\beta)}(Y_2) \Bigl( P_{(\alpha)}(Y_1) - P_{(\beta)}(Y_1) \Bigr) + P_{(\beta)}(Y_1) \Bigl( P_{(\alpha)}(Y_2) - P_{(\beta)}(Y_2) \Bigr) \\
&+ \Bigl( P_{(\alpha\beta\gamma)}(Y_2)+P_{(\alpha\beta\gamma\delta)}(Y_2) \Bigr) \Bigl( P_{(\alpha\beta)}(Y_1) - P_{(\beta\gamma)}(Y_1) \Bigr) \\
&+ \Bigl( P_{(\alpha\beta\gamma)}(Y_1)+P_{(\alpha\beta\gamma\delta)}(Y_1) \Bigr) \Bigl( P_{(\alpha\beta)}(Y_2) - P_{(\beta\gamma)}(Y_2) \Bigr) \\
&+ P_{(\alpha\beta\gamma)}(Y_2) \Bigl( P_{(\alpha\beta\gamma)}(Y_1) - P_{(\beta\gamma\delta)}(Y_1) \Bigr) + P_{(\alpha\beta\gamma)}(Y_1) \Bigl( P_{(\alpha\beta\gamma)}(Y_2) - P_{(\beta\gamma\delta)}(Y_2) \Bigr)
\end{align*}
\begin{align*}
=& P_{(\alpha\beta)}(Y_1) P_{(\alpha\beta)}(Y_2) - P_{(\beta\gamma)}(Y_1) P_{(\beta\gamma)}(Y_2) \\
&+ P_{(\alpha\beta)}(Y_1) P_{(\alpha\beta\gamma)}(Y_2) - P_{(\beta\gamma)}(Y_1) P_{(\beta\gamma\delta)}(Y_2) \\
&+ P_{(\alpha\beta\gamma)}(Y_1) P_{(\alpha\beta)}(Y_2) - P_{(\beta\gamma\delta)}(Y_1) P_{(\beta\gamma)}(Y_2) \\
&+P_{(\beta)}(Y_2) P_1 \Bigl( P_{\alpha}(Y_1) - P_{\beta}(Y_1) \Bigr) + P_{(\beta)}(Y_1) P_2 \Bigl( P_{\alpha}(Y_2) - P_{\beta}(Y_2) \Bigr) \\
&+ \Bigl( P_{(\alpha\beta\gamma)}(Y_2)+P_{(\alpha\beta\gamma\delta)}(Y_2) \Bigr) P_1 \Bigl( P_{\alpha\beta}(Y_1) - P_{\beta\gamma}(Y_1) \Bigr) \\
&+ \Bigl( P_{(\alpha\beta\gamma)}(Y_1)+P_{(\alpha\beta\gamma\delta)}(Y_1) \Bigr) P_2 \Bigl( P_{\alpha\beta}(Y_2) - P_{\beta\gamma}(Y_2) \Bigr) \\
&+ P_{(\alpha\beta\gamma)}(Y_2) P_1 \Bigl( P_{\alpha\beta\gamma}(Y_1) - P_{\beta\gamma\delta}(Y_1) \Bigr) + P_{(\alpha\beta\gamma)}(Y_1) P_2 \Bigl( P_{\alpha\beta\gamma}(Y_2) - P_{\beta\gamma\delta}(Y_2) \Bigr) \\
&\text{by } \eqref{4(a)(b)}, \eqref{4(ab)(bc)} \text{ and } \eqref{4(abc)(bcd)} 
.
\end{align*} Again by \eqref{4(a)(a)(b)(b)} and the inductive assumption this term is non-negative, and therefore concludes the proof for $P_{\alpha\beta}(X) \geq P_{\beta\gamma}(X)$.
Moreover we have
\begin{align*}
&P_{\alpha\beta\gamma}(X) - P_{\beta\gamma\delta}(X) \\
=&P_{(\alpha)}(Y_1) P_{(\beta\gamma)}(Y_2) + P_{(\beta\gamma)}(Y_1) P_{(\alpha)}(Y_2)+ 2P_{(\beta)}(Y_1) P_{(\alpha\beta)}(Y_2) \\
&+ 2P_{(\alpha\beta)}(Y_1) P_{(\beta)}(Y_2) + P_{(\alpha\beta\gamma)}(Y_1) P_{(\alpha\beta\gamma)}(Y_2) + P_{(\alpha\beta\gamma)}(Y_1) P_{(\alpha\beta\gamma\delta)}(Y_2) \\
&+P_{(\alpha\beta\gamma\delta)}(Y_1) P_{(\alpha\beta\gamma)}(Y_2) - 3P_{(\beta)}(Y_1) P_{(\beta\gamma)}(Y_2) - 3P_{(\beta\gamma)}(Y_1) P_{(\beta)}(Y_2) \\
&- P_{(\beta\gamma\delta)}(Y_1) P_{(\beta\gamma\delta)}(Y_2) - P_{(\beta\gamma\delta)}(Y_1) P_{(\alpha\beta\gamma\delta)}(Y_2) - P_{(\alpha\beta\gamma\delta)}(Y_1) P_{(\beta\gamma\delta)}(Y_2) \\
&\text{by } \eqref{alphabetagamma} \text{ and } \eqref{betagammadelta} \\
=& P_{(\beta\gamma)}(Y_2) \Bigl( P_{(\alpha)}(Y_1) - P_{(\beta)}(Y_1) \Bigr) + P_{(\beta\gamma)}(Y_1) \Bigl( P_{(\alpha)}(Y_2) - P_{(\beta)}(Y_2) \Bigr) \\
&+2P_{(\beta)}(Y_2) \Bigl( P_{(\alpha\beta)}(Y_1) - P_{(\beta\gamma)}(Y_1) \Bigr) + 2P_{(\beta)}(Y_1) \Bigl( P_{(\alpha\beta)}(Y_2) - P_{(\beta\gamma)}(Y_2) \Bigr) \\
&+P_{(\alpha\beta\gamma)}(Y_1) P_{(\alpha\beta\gamma)}(Y_2) - P_{(\beta\gamma\delta)}(Y_1) P_{(\beta\gamma\delta)}(Y_2) \\
&+ P_{(\alpha\beta\gamma\delta)}(Y_2) \Bigl( P_{(\alpha\beta\gamma)}(Y_1) - P_{(\beta\gamma\delta)}(Y_1) \Bigr) + P_{(\alpha\beta\gamma\delta)}(Y_1) \Bigl( P_{(\alpha\beta\gamma)}(Y_2) - P_{(\beta\gamma\delta)}(Y_2) \Bigr) \\
=& P_{(\beta\gamma)}(Y_2) P_1 \Bigl( P_{\alpha}(Y_1) - P_{\beta}(Y_1) \Bigr) + P_{(\beta\gamma)}(Y_1)P_2 \Bigl( P_{\alpha}(Y_2) - P_{\beta}(Y_2) \Bigr) \\
&+2P_{(\beta)}(Y_2) P_1 \Bigl( P_{\alpha\beta}(Y_1) - P_{\beta\gamma}(Y_1) \Bigr) + 2P_{(\beta)}(Y_1) P_2 \Bigl( P_{\alpha\beta}(Y_2) - P_{\beta\gamma}(Y_2) \Bigr) \\
&+P_{(\alpha\beta\gamma)}(Y_1) P_{(\alpha\beta\gamma)}(Y_2) - P_{(\beta\gamma\delta)}(Y_1) P_{(\beta\gamma\delta)}(Y_2) \\
&+ P_{(\alpha\beta\gamma\delta)}(Y_2) P_1 \Bigl( P_{\alpha\beta\gamma}(Y_1) - P_{\beta\gamma\delta}(Y_1) \Bigr) + P_{(\alpha\beta\gamma\delta)}(Y_1) P_2 \Bigl( P_{\alpha\beta\gamma}(Y_2) - P_{\beta\gamma\delta}(Y_2) \Bigr) \\
&\text{by } \eqref{4(a)(b)}, \eqref{4(ab)(bc)} \text{ and } \eqref{4(abc)(bcd)} 
.
\end{align*}
By \eqref{4(a)(a)(b)(b)} and the inductive assumption $P_{\alpha\beta\gamma}(X) - P_{\beta\gamma\delta}(X)$ is non-negative, and therefore concludes the proof of the last part of Lemma \ref{lemmageq}. \qed
\\
\noindent
\textbf{Extension to the proof of Theorem \ref{theoremRA}}
First of all we state some equations for $i \in \{1,2\}$ which helps to show \eqref{8D(X)}.
\begin{align}
1=&P_{\alpha}(Y_i)+P_{\beta}(Y_i)+P_{\gamma}(Y_i)+P_{\delta}(Y_i)+P_{\alpha\beta}(Y_i)+P_{\alpha\gamma}(Y_i)+P_{\alpha\delta}(Y_i)+P_{\beta\gamma}(Y_i) \nonumber
\\&+P_{\beta\delta}(Y_i)+P_{\gamma\delta}(Y_i)+P_{\alpha\beta\gamma}(Y_i)+P_{\alpha\beta\delta}(Y_i)+P_{\alpha\gamma\delta}(Y_i)+P_{\beta\gamma\delta}(Y_i)+P_{\alpha\beta\gamma\delta}(Y_i) \nonumber \\
=&P_{\alpha}(Y_i)+3P_{\beta}(Y_i)+3P_{\alpha\beta}(Y_i)+3P_{\beta\gamma}(Y_i)+3P_{\alpha\beta\gamma}(Y_i)+P_{\beta\gamma\delta}(Y_i)+P_{\alpha\beta\gamma\delta}(Y_i) \label{sum} \\
&\text{by } \eqref{x1}, \eqref{x2}, \eqref{x3}. \nonumber
\end{align}
By \eqref{sum} we have that
\begin{align}
3P_{\beta}(Y_i)=1-P_{\alpha}(Y_i)-3P_{\alpha\beta}(Y_i)-3P_{\beta\gamma}(Y_i)-3P_{\alpha\beta\gamma}(Y_i)-P_{\beta\gamma\delta}(Y_i)-P_{\alpha\beta\gamma\delta}(Y_i) \label{beta}
\end{align}
and
\begin{align}
&P_{(\alpha)}(Y_i)+3P_{(\beta)}(Y_i)+3P_{(\alpha\beta)}(Y_i)+3P_{(\beta\gamma)}(Y_i)+3P_{(\alpha\beta\gamma)}(Y_i)+P_{(\beta\gamma\delta)}(Y_i)+P_{(\alpha\beta\gamma\delta)}(Y_i) \nonumber \\
&=(1-3p_i) P_{\alpha}(Y_i) + 3p_i P_{\beta}(Y_i)+ 3(1-p_i) P_{\beta}(Y_i) + 3p_i P_{\alpha}(Y_i) \nonumber \\
&\quad + 3(1-2p_i) P_{\alpha\beta}(Y_i) + 6p_i P_{\beta\gamma}(Y_i) +3(1-2p_i) P_{\beta\gamma}(Y_i) + 6p_i P_{\alpha\beta}(Y_i) \nonumber\\
&\quad + 3(1-p_i) P_{\alpha\beta\gamma}(Y_i) + 3p_i P_{\beta\gamma\delta}(Y_i) +(1-3p_i) P_{\beta\gamma\delta}(Y_i) + 3p_i P_{\alpha\beta\gamma}(Y_i) + P_{\alpha\beta\gamma\delta}(Y_i) \nonumber \\
&\qquad \text{by } \eqref{w1}, \eqref{w2}, \eqref{w3}, \eqref{w4}, \eqref{w5}, \eqref{w6}, \eqref{w7} \nonumber \\
&=P_{\alpha}(Y_i)+3P_{\beta}(Y_i)+3P_{\alpha\beta}(Y_i)+3P_{\beta\gamma}(Y_i)+3P_{\alpha\beta\gamma}(Y_i)+P_{\beta\gamma\delta}(Y_i)+P_{\alpha\beta\gamma\delta}(Y_i) \nonumber \\
&=1 \label{sum2}.
\end{align} Furthermore, the following expressions can be simplified by \eqref{RA4}, \eqref{beta} \text{ and } \eqref{sum2}.
\begin{align}
&4P_{(\alpha)}(Y_i) + 6P_{(\alpha\beta)}(Y_i) + 4P_{(\alpha\beta\gamma)}(Y_i) + P_{(\alpha\beta\gamma\delta)}(Y_i) \nonumber \\
= & 4 \Bigl( P_{(\alpha)}(Y_i) + \frac{3}{2} P_{(\alpha\beta)}(Y_i) + P_{(\alpha\beta\gamma)}(Y_i) + \frac{1}{4} P_{(\alpha\beta\gamma\delta)}(Y_i) \Bigr) \nonumber \\
= & 4 \Bigl( (1-3p_i) P_{\alpha}(Y_i) + 3p_i P_{\beta}(Y_i)) + \frac{3}{2} ((1-2p_i) P_{\alpha\beta}(Y_i) + 2p_i P_{\beta\gamma}(Y_i)) \nonumber \\
& + (1-p_i) P_{\alpha\beta\gamma}(Y_i) + p_i P_{\beta\gamma\delta}(Y_i) + \frac{1}{4} P_{\alpha\beta\gamma\delta}(Y_i) \Bigr) \nonumber \\
=&4 \Bigl( (1-3p_i) P_{\alpha}(Y_i) + p_i  \bigl( 1-P_{\alpha}(Y_i)-3P_{\alpha\beta}(Y_i)-3P_{\beta\gamma}(Y_i)-3P_{\alpha\beta\gamma}(Y_i)-P_{\beta\gamma\delta}(Y_i)-P_{\alpha\beta\gamma\delta}(Y_i) \bigr) \nonumber \\
& + \frac{3}{2} ((1-2p_i) P_{\alpha\beta}(Y_i) + 2p_i P_{\beta\gamma}(Y_i)) + (1-p_i) P_{\alpha\beta\gamma}(Y_i) + p_i P_{\beta\gamma\delta}(Y_i) + \frac{1}{4} P_{\alpha\beta\gamma\delta}(Y_i) \Bigr) \nonumber \\
&\text{by } \eqref{beta} \nonumber \\
=&4 \Bigl( (1-3p_i) P_{\alpha}(Y_i) + p_i -p_i P_{\alpha}(Y_i)- 3p_i P_{\alpha\beta}(Y_i)- 3p_i P_{\beta\gamma}(Y_i)- 3p_i P_{\alpha\beta\gamma}(Y_i)- p_i P_{\beta\gamma\delta}(Y_i) \nonumber \\
& - p_i P_{\alpha\beta\gamma\delta}(Y_i)  + \frac{3}{2} ((1-2p_i) P_{\alpha\beta}(Y_i) + 2p_i P_{\beta\gamma}(Y_i)) + (1-p_i) P_{\alpha\beta\gamma}(Y_i) + p_i P_{\beta\gamma\delta}(Y_i) + \frac{1}{4} P_{\alpha\beta\gamma\delta}(Y_i) \Bigr) \nonumber \\
=&4 \Bigl( (1-3p_i) P_{\alpha}(Y_i) + p_i -p_i P_{\alpha}(Y_i)- 3p_i P_{\alpha\beta}(Y_i) - 3p_i P_{\alpha\beta\gamma}(Y_i) \nonumber \\
& - p_i P_{\alpha\beta\gamma\delta}(Y_i)  + \frac{3}{2} (1-2p_i) P_{\alpha\beta}(Y_i) + (1-p_i) P_{\alpha\beta\gamma}(Y_i) + \frac{1}{4} P_{\alpha\beta\gamma\delta}(Y_i) \Bigr) \nonumber \\
=&4 \Bigl( p_i + (1-4p_i) P_{\alpha}(Y_i) + \frac{3}{2} (1-4p_i) P_{\alpha\beta}(Y_i) + (1-4p_i) P_{\alpha\beta\gamma}(Y_i) + \frac{1}{4} (1-4p_i) P_{\alpha\beta\gamma\delta}(Y_i) \Bigr) \nonumber \\
=&4 \Bigl( p_i + (1-4p_i) \Bigl( P_{\alpha}(Y_i) + \frac{3}{2} P_{\alpha\beta}(Y_i) + P_{\alpha\beta\gamma}(Y_i) + \frac{1}{4} P_{\alpha\beta\gamma\delta}(Y_i) \Bigr) \Bigr) \nonumber \\
=&4 \Bigl( p_i + (1-4p_i) RA(Y_i) \Bigr) \label{piRA} \\
&\text{by } \eqref{RA4} \nonumber
.
\end{align} Moreover,
\begin{align}
& \frac{1}{2}P_{(\alpha)}(Y_i) + \frac{3}{2} P_{(\alpha\beta)}(Y_i) + \frac{3}{2} P_{(\alpha\beta\gamma)}(Y_i) + P_{(\alpha\beta\gamma\delta)}(Y_i) + \frac{3}{2} P_{(\beta)}(Y_i) + P_{(\beta\gamma)}(Y_i) + \frac{1}{4} P_{(\beta\gamma\delta)}(Y_i) \nonumber \\
&= \frac{1}{2} P_{(\alpha)}(Y_i) + \frac{3}{2} P_{(\beta)}(Y_i) + \frac{3}{2} P_{(\alpha\beta)}(Y_i) + \frac{3}{2} P_{(\beta\gamma)}(Y_i) - \frac{1}{2} P_{(\beta\gamma)}(Y_i) + \frac{3}{2} P_{(\alpha\beta\gamma)}(Y_i) \nonumber \\
& \quad + \frac{1}{2} P_{(\beta\gamma\delta)}(Y_i) - \frac{1}{4} P_{(\beta\gamma\delta)}(Y_i) + \frac{1}{2} P_{(\alpha\beta\gamma\delta)}(Y_i) + \frac{1}{2} P_{(\alpha\beta\gamma\delta)}(Y_i) \nonumber \\
&= \frac{1}{2} \Bigl( P_{(\alpha)}(Y_i) + 3P_{(\beta)}(Y_i) + 3P_{(\alpha\beta)}(Y_i) + 3P_{(\beta\gamma)}(Y_i) + 3P_{(\alpha\beta\gamma)}(Y_i) + P_{(\beta\gamma\delta)}(Y_i) + P_{(\alpha\beta\gamma\delta)}(Y_i) \Bigr) \nonumber \\
&\quad - \frac{1}{2} P_{(\beta\gamma)}(Y_i) - \frac{1}{4} P_{(\beta\gamma\delta)}(Y_i) + \frac{1}{2} P_{(\alpha\beta\gamma\delta)}(Y_i) \nonumber \\
&=\frac{1}{2} - \frac{1}{2} P_{(\beta\gamma)}(Y_i) - \frac{1}{4} P_{(\beta\gamma\delta)}(Y_i) + \frac{1}{2} P_{(\alpha\beta\gamma\delta)}(Y_i) \label{beta1} \\
&\quad \text{by } \eqref{sum2}. \nonumber
\end{align} Additionally,
\begin{align}
&\frac{3}{2} P_{(\alpha)}(Y_i) + 2 P_{(\beta)}(Y_i) + \frac{15}{4} P_{(\alpha\beta)}(Y_i) + \frac{3}{4} P_{(\beta\gamma)}(Y_i) + 3P_{(\alpha\beta\gamma)}(Y_i) + \frac{3}{2} P_{(\alpha\beta\gamma\delta)}(Y_i) \nonumber \\
&= \frac{2}{3} \Bigl( P_{(\alpha)}(Y_i) + 3P_{(\beta)}(Y_i) + 3P_{(\alpha\beta)}(Y_i) + 3P_{(\beta\gamma)}(Y_i) + 3P_{(\alpha\beta\gamma)}(Y_i) + P_{(\beta\gamma\delta)}(Y_i) + P_{(\alpha\beta\gamma\delta)}(Y_i) \Bigr) \nonumber \\
&\quad + \frac{5}{6} P_{(\alpha)}(Y_i) + \frac{7}{4} P_{(\alpha\beta)}(Y_i) - \frac{5}{4} P_{(\beta\gamma)}(Y_i) + P_{(\alpha\beta\gamma)}(Y_i) - \frac{2}{3} P_{(\beta\gamma\delta)}(Y_i) + \frac{5}{6} P_{(\alpha\beta\gamma\delta)}(Y_i) \nonumber \\
&= \frac{2}{3} + \frac{5}{6} P_{(\alpha)}(Y_i) + \frac{7}{4} P_{(\alpha\beta)}(Y_i) - \frac{5}{4} P_{(\beta\gamma)}(Y_i) + P_{(\alpha\beta\gamma)}(Y_i) - \frac{2}{3} P_{(\beta\gamma\delta)}(Y_i) + \frac{5}{6} P_{(\alpha\beta\gamma\delta)}(Y_i) \label{beta2} \\
&\quad \text{by } \eqref{sum2}, \nonumber
\end{align} and
\begin{align}
&\frac{3}{2} P_{(\alpha)}(Y_i) + \frac{3}{4} P_{(\beta)}(Y_i) + 3P_{(\alpha\beta)}(Y_i) + 2P_{(\alpha\beta\gamma)}(Y_i) + P_{(\alpha\beta\gamma\delta)}(Y_i) \nonumber \\
&= \frac{1}{4} \Bigl( P_{(\alpha)}(Y_i) + 3P_{(\beta)}(Y_i) + 3P_{(\alpha\beta)}(Y_i) + 3P_{(\beta\gamma)}(Y_i) + 3P_{(\alpha\beta\gamma)}(Y_i) + P_{(\beta\gamma\delta)}(Y_i) + P_{(\alpha\beta\gamma\delta)}(Y_i) \Bigr) \nonumber \\
&\quad + \frac{5}{4} P_{(\alpha)}(Y_i) + \frac{9}{4} P_{(\alpha\beta)}(Y_i) - \frac{3}{4} P_{(\beta\gamma)}(Y_i) + \frac{5}{4} P_{(\alpha\beta\gamma)}(Y_i) - \frac{1}{4} P_{(\beta\gamma\delta)}(Y_i) + \frac{3}{4} P_{(\alpha\beta\gamma\delta)}(Y_i) \nonumber \\
&=\frac{1}{4} + \frac{5}{4} P_{(\alpha)}(Y_i) + \frac{9}{4} P_{(\alpha\beta)}(Y_i) - \frac{3}{4} P_{(\beta\gamma)}(Y_i) + \frac{5}{4} P_{(\alpha\beta\gamma)}(Y_i) - \frac{1}{4} P_{(\beta\gamma\delta)}(Y_i) + \frac{3}{4} P_{(\alpha\beta\gamma\delta)}(Y_i) \label{beta3} \\
&\quad \text{by } \eqref{sum2} \nonumber. 
\end{align} Furthermore,
\begin{align}
4p_i-1 &= 4p_i - 4 + 3 +12p -12p \nonumber \\
&= 3 - 12p + 4p_i - 4 +12(p_i + p_i^{'} -4 p_i p_i^{'}) \nonumber \\
&\quad \text{by } \eqref{defp} \nonumber \\
&= 3(1-4p) + 16p_i - 4 + 12p_i^{'} -48 p_i p_i^{'} \nonumber \\
&= 3(1-4p) + 4(1-4p_i)(-1 + 3p_i^{'}) \nonumber \\
&= 3P + 4P_i(-1 + 3p_i^{'}). \label{4pi1} \\ 
&\quad \text{by the definition of } P \text{ and } P_i, \nonumber 
\end{align} and
\begin{align}
4p_i-4+12p &= 4p_i-4+12(p_i + p_i^{'} -4 p_i p_i^{'}) \nonumber \\
&\quad \text{by } \eqref{defp} \nonumber \\
&= 4 (4p_i -1 + 3p_i^{'} - 12 p_i p_i^{'}) \nonumber \\
&= 4(1-4p_i)(-1 + 3p_i^{'}) \nonumber \\
&= 4P_i(-1 + 3p_i^{'}) \label{4pi2} \\
&\quad \text{by the definition of } P_i \nonumber
.
\end{align}
By using the simplifications stated before we can now rewrite $RA(X)$.
\begin{align*}
RA(X) = & P_{\alpha}(X)+\frac{3}{2} \cdot P_{\alpha\beta}(X) + P_{\alpha\beta\gamma}(X) + \frac{1}{4} \cdot P_{\alpha\beta\gamma\delta}(X) \\
&\text{by } \eqref{RA4} \\
= &P_{(\alpha)}(Y_1) P_{(\alpha)}(Y_2) + 3P_{(\alpha)}(Y_1) P_{(\alpha\beta)}(Y_2) + 3P_{(\alpha\beta)}(Y_1)  P_{(\alpha)}(Y_2) \\
&+3P_{(\alpha)}(Y_1) P_{(\alpha\beta\gamma)}(Y_2) + 3P_{(\alpha\beta\gamma)}(Y_1) P_{(\alpha)}(Y_2) + 6P_{(\alpha\beta)}(Y_1) P_{(\alpha\beta)}(Y_2) \\
&+3P_{(\alpha\beta)}(Y_1) P_{(\alpha\beta\gamma)}(Y_2) + 3P_{(\alpha\beta\gamma)}(Y_1) P_{(\alpha\beta)}(Y_2) + P_{(\alpha)}(Y_1) P_{(\alpha\beta\gamma\delta)}(Y_2)\\
&+ P_{(\alpha\beta\gamma\delta)}(Y_1) P_{(\alpha)}(Y_2)+\frac{3}{2} \Bigl( P_{(\alpha)}(Y_1) P_{(\beta)}(Y_2) + P_{(\beta)}(Y_1) P_{(\alpha)}(Y_2) \\
&+P_{(\alpha\beta)}(Y_1) P_{(\alpha\beta)}(Y_2) + 2P_{(\alpha\beta)}(Y_1) P_{(\alpha\beta\gamma)}(Y_2) + 2P_{(\alpha\beta\gamma)}(Y_1) P_{(\alpha\beta)}(Y_2)\\
&+P_{(\alpha\beta)}(Y_1) P_{(\alpha\beta\gamma\delta)}(Y_2) + P_{(\alpha\beta\gamma\delta)}(Y_1) P_{(\alpha\beta)}(Y_2) +2P_{(\alpha\beta\gamma)}(Y_1) P_{(\alpha\beta\gamma)}(Y_2)  \Bigr) \\
&+P_{(\alpha)}(Y_1) P_{(\beta\gamma)}(Y_2) + P_{(\beta\gamma)}(Y_1) P_{(\alpha)}(Y_2)+ 2P_{(\beta)}(Y_1) P_{(\alpha\beta)}(Y_2) \\
&+ 2P_{(\alpha\beta)}(Y_1) P_{(\beta)}(Y_2)+P_{(\alpha\beta\gamma)}(Y_1) P_{(\alpha\beta\gamma)}(Y_2) + P_{(\alpha\beta\gamma)}(Y_1) P_{(\alpha\beta\gamma\delta)}(Y_2) \\
&+ P_{(\alpha\beta\gamma\delta)}(Y_1) P_{(\alpha\beta\gamma)}(Y_2) +\frac{1}{4} \Bigl( P_{(\alpha)}(Y_1) P_{(\beta\gamma\delta)}(Y_2) + P_{(\beta\gamma\delta)}(Y_1) P_{(\alpha)}(Y_2) \\
&+3P_{(\beta)}(Y_1) P_{(\alpha\beta\gamma)}(Y_2) + 3P_{(\alpha\beta\gamma)}(Y_1) P_{(\beta)}(Y_2)+ 3P_{(\alpha\beta)}(Y_1) P_{(\beta\gamma)}(Y_2) \\
&+3 P_{(\beta\gamma)}(Y_1) P_{(\alpha\beta)}(Y_2) + P_{(\alpha\beta\gamma\delta)}(Y_1) P_{(\alpha\beta\gamma\delta)}(Y_2) \Bigr) \\
&\text{by } \eqref{alpha}, \eqref{alphabeta}, \eqref{alphabetagamma}, \eqref{alphabetagammadelta}
\end{align*}
\begin{align*}
= &P_{(\alpha)}(Y_1) \Bigl( \frac{1}{2}P_{(\alpha)}(Y_2) + \frac{3}{2} P_{(\alpha\beta)}(Y_2) + \frac{3}{2} P_{(\alpha\beta\gamma)}(Y_2) + P_{(\alpha\beta\gamma\delta)}(Y_2) + \frac{3}{2} P_{(\beta)}(Y_2) + P_{(\beta\gamma)}(Y_2)\\
&\qquad \qquad + \frac{1}{4} P_{(\beta\gamma\delta)}(Y_2) \Bigr) \\
&+P_{(\alpha)}(Y_2) \Bigl( \frac{1}{2}P_{(\alpha)}(Y_1) + \frac{3}{2} P_{(\alpha\beta)}(Y_1) + \frac{3}{2} P_{(\alpha\beta\gamma)}(Y_1) + P_{(\alpha\beta\gamma\delta)}(Y_1) + \frac{3}{2} P_{(\beta)}(Y_1) + P_{(\beta\gamma)}(Y_1) \\
&\qquad \qquad + \frac{1}{4} P_{(\beta\gamma\delta)}(Y_1) \Bigr) \\
&+P_{(\alpha\beta)}(Y_1) \Bigl( \frac{3}{2} P_{(\alpha)}(Y_2) + 2 P_{(\beta)}(Y_2) + \frac{15}{4} P_{(\alpha\beta)}(Y_2) + \frac{3}{4} P_{(\beta\gamma)}(Y_2) + 3P_{(\alpha\beta\gamma)}(Y_2) + \frac{3}{2} P_{(\alpha\beta\gamma\delta)}(Y_2) \Bigr) \\
&+P_{(\alpha\beta)}(Y_2) \Bigl( \frac{3}{2} P_{(\alpha)}(Y_1) + 2 P_{(\beta)}(Y_1) + \frac{15}{4} P_{(\alpha\beta)}(Y_1) + \frac{3}{4} P_{(\beta\gamma)}(Y_1) + 3P_{(\alpha\beta\gamma)}(Y_1) + \frac{3}{2} P_{(\alpha\beta\gamma\delta)}(Y_1) \Bigr) \\
&+P_{(\alpha\beta\gamma)}(Y_1) \Bigl(\frac{3}{2} P_{(\alpha)}(Y_2)+ \frac{3}{4} P_{(\beta)}(Y_2) + 3 P_{(\alpha\beta)}(Y_2) + 2 P_{(\alpha\beta\gamma)}(Y_2) + P_{(\alpha\beta\gamma\delta)}(Y_2) \Bigr) \\
&+P_{(\alpha\beta\gamma)}(Y_2) \Bigl( \frac{3}{2} P_{(\alpha)}(Y_1)+ \frac{3}{4} P_{(\beta)}(Y_1) + 3 P_{(\alpha\beta)}(Y_1) + 2 P_{(\alpha\beta\gamma)}(Y_1) + P_{(\alpha\beta\gamma\delta)}(Y_1) \Bigr) \\
&+ \frac{1}{4} P_{(\alpha\beta\gamma\delta)}(Y_1)P_{(\alpha\beta\gamma\delta)}(Y_2) \\
= &P_{(\alpha)}(Y_1) \Bigl( \frac{1}{2} - \frac{1}{2} P_{(\beta\gamma)}(Y_2) - \frac{1}{4} P_{(\beta\gamma\delta)}(Y_2) + \frac{1}{2} P_{(\alpha\beta\gamma\delta)}(Y_2) \Bigr) \\
&+P_{(\alpha)}(Y_2) \Bigl( \frac{1}{2} - \frac{1}{2} P_{(\beta\gamma)}(Y_1) - \frac{1}{4} P_{(\beta\gamma\delta)}(Y_1) + \frac{1}{2} P_{(\alpha\beta\gamma\delta)}(Y_1) \Bigr) \\
&+P_{(\alpha\beta)}(Y_1) \Bigl(\frac{2}{3} + \frac{5}{6} P_{(\alpha)}(Y_2) + \frac{7}{4} P_{(\alpha\beta)}(Y_2) - \frac{5}{4} P_{(\beta\gamma)}(Y_2) + P_{(\alpha\beta\gamma)}(Y_2) - \frac{2}{3} P_{(\beta\gamma\delta)}(Y_2) + \frac{5}{6} P_{(\alpha\beta\gamma\delta)}(Y_2) \Bigr) \\
&+P_{(\alpha\beta)}(Y_2) \Bigl(\frac{2}{3} + \frac{5}{6} P_{(\alpha)}(Y_1) + \frac{7}{4} P_{(\alpha\beta)}(Y_1) - \frac{5}{4} P_{(\beta\gamma)}(Y_1) + P_{(\alpha\beta\gamma)}(Y_1) - \frac{2}{3} P_{(\beta\gamma\delta)}(Y_1) + \frac{5}{6} P_{(\alpha\beta\gamma\delta)}(Y_1) \Bigr) \\
&+P_{(\alpha\beta\gamma)}(Y_1) \Bigl(\frac{1}{4} + \frac{5}{4} P_{(\alpha)}(Y_2) + \frac{9}{4} P_{(\alpha\beta)}(Y_2) - \frac{3}{4} P_{(\beta\gamma)}(Y_2) + \frac{5}{4} P_{(\alpha\beta\gamma)}(Y_2) - \frac{1}{4} P_{(\beta\gamma\delta)}(Y_2) + \frac{3}{4} P_{(\alpha\beta\gamma\delta)}(Y_2) \Bigr) \\
&+P_{(\alpha\beta\gamma)}(Y_2) \Bigl(\frac{1}{4} + \frac{5}{4} P_{(\alpha)}(Y_1) + \frac{9}{4} P_{(\alpha\beta)}(Y_1) - \frac{3}{4} P_{(\beta\gamma)}(Y_1) + \frac{5}{4} P_{(\alpha\beta\gamma)}(Y_1) - \frac{1}{4} P_{(\beta\gamma\delta)}(Y_1) + \frac{3}{4} P_{(\alpha\beta\gamma\delta)}(Y_1) \Bigr) \\
&+ \frac{1}{4} P_{(\alpha\beta\gamma\delta)}(Y_1)P_{(\alpha\beta\gamma\delta)}(Y_2) \\
&\text{by } \eqref{beta1}, \eqref{beta2}, \eqref{beta3} \\
\end{align*}

\begin{align*}
= & \Bigl( \frac{1}{2}P_{(\alpha)}(Y_1)+ \frac{5}{4}P_{(\alpha\beta)}(Y_1)+ \frac{3}{4}P_{(\alpha\beta\gamma)}(Y_1) \Bigr) \Bigl( P_{(\alpha\beta)}(Y_2)-P_{(\beta\gamma)}(Y_2) \Bigr)\\
&+ \Bigl( \frac{1}{2}P_{(\alpha)}(Y_2)+ \frac{5}{4}P_{(\alpha\beta)}(Y_2)+ \frac{3}{4}P_{(\alpha\beta\gamma)}(Y_2) \Bigr) \Bigl( P_{(\alpha\beta)}(Y_1)-P_{(\beta\gamma)}(Y_1) \Bigr)\\
&+ \Bigl( \frac{1}{4}P_{(\alpha)}(Y_1)+ \frac{2}{3}P_{(\alpha\beta)}(Y_1)+ \frac{1}{4}P_{(\alpha\beta\gamma)}(Y_1) \Bigr) \Bigl( P_{(\alpha\beta\gamma)}(Y_2)-P_{(\beta\gamma\delta)}(Y_2) \Bigr)\\
&+ \Bigl( \frac{1}{4}P_{(\alpha)}(Y_2)+ \frac{2}{3}P_{(\alpha\beta)}(Y_2)+ \frac{1}{4}P_{(\alpha\beta\gamma)}(Y_2) \Bigr) \Bigl( P_{(\alpha\beta\gamma)}(Y_1)-P_{(\beta\gamma\delta)}(Y_1) \Bigr)\\
&+P_{(\alpha)}(Y_1) \Bigl( \frac{1}{2} + \frac{1}{2} P_{(\alpha\beta\gamma\delta)}(Y_2) \Bigr) +P_{(\alpha)}(Y_2) \Bigl( \frac{1}{2} + \frac{1}{2} P_{(\alpha\beta\gamma\delta)}(Y_1) \Bigr) \\
&+P_{(\alpha\beta)}(Y_1) \Bigl(\frac{2}{3} + \frac{1}{3} P_{(\alpha)}(Y_2) + \frac{1}{2} P_{(\alpha\beta)}(Y_2)  + \frac{1}{3} P_{(\alpha\beta\gamma)}(Y_2) + \frac{5}{6} P_{(\alpha\beta\gamma\delta)}(Y_2) \Bigr) \\
&+P_{(\alpha\beta)}(Y_2) \Bigl(\frac{2}{3} + \frac{1}{3} P_{(\alpha)}(Y_1) + \frac{1}{2} P_{(\alpha\beta)}(Y_1) + \frac{1}{3} P_{(\alpha\beta\gamma)}(Y_1) + \frac{5}{6} P_{(\alpha\beta\gamma\delta)}(Y_1) \Bigr) \\
&+P_{(\alpha\beta\gamma)}(Y_1) \Bigl(\frac{1}{4} + P_{(\alpha)}(Y_2) + \frac{3}{2} P_{(\alpha\beta)}(Y_2) + P_{(\alpha\beta\gamma)}(Y_2) + \frac{3}{4} P_{(\alpha\beta\gamma\delta)}(Y_2) \Bigr) \\
&+P_{(\alpha\beta\gamma)}(Y_2) \Bigl(\frac{1}{4} + P_{(\alpha)}(Y_1) + \frac{3}{2} P_{(\alpha\beta)}(Y_1) + P_{(\alpha\beta\gamma)}(Y_1) + \frac{3}{4} P_{(\alpha\beta\gamma\delta)}(Y_1) \Bigr) \\
&+ \frac{1}{4} P_{(\alpha\beta\gamma\delta)}(Y_1)P_{(\alpha\beta\gamma\delta)}(Y_2) 
\end{align*}
\begin{align*}
= & \Bigl( \frac{1}{2}P_{(\alpha)}(Y_1)+ \frac{5}{4}P_{(\alpha\beta)}(Y_1)+ \frac{3}{4}P_{(\alpha\beta\gamma)}(Y_1) \Bigr) \Bigl( P_{(\alpha\beta)}(Y_2)-P_{(\beta\gamma)}(Y_2) \Bigr)\\
&+ \Bigl( \frac{1}{2}P_{(\alpha)}(Y_2)+ \frac{5}{4}P_{(\alpha\beta)}(Y_2)+ \frac{3}{4}P_{(\alpha\beta\gamma)}(Y_2) \Bigr) \Bigl( P_{(\alpha\beta)}(Y_1)-P_{(\beta\gamma)}(Y_1) \Bigr)\\
&+ \Bigl( \frac{1}{4}P_{(\alpha)}(Y_1)+ \frac{2}{3}P_{(\alpha\beta)}(Y_1)+ \frac{1}{4}P_{(\alpha\beta\gamma)}(Y_1) \Bigr) \Bigl( P_{(\alpha\beta\gamma)}(Y_2)-P_{(\beta\gamma\delta)}(Y_2) \Bigr)\\
&+ \Bigl( \frac{1}{4}P_{(\alpha)}(Y_2)+ \frac{2}{3}P_{(\alpha\beta)}(Y_2)+ \frac{1}{4}P_{(\alpha\beta\gamma)}(Y_2) \Bigr) \Bigl( P_{(\alpha\beta\gamma)}(Y_1)-P_{(\beta\gamma\delta)}(Y_1) \Bigr)\\
&+\frac{1}{2} P_{(\alpha)}(Y_1) + \frac{3}{4} P_{(\alpha\beta)}(Y_1) + \frac{1}{2} P_{(\alpha\beta\gamma)}(Y_1) + \frac{1}{8} P_{(\alpha\beta\gamma\delta)}(Y_1)\\
&+\frac{1}{2} P_{(\alpha)}(Y_2) + \frac{3}{4} P_{(\alpha\beta)}(Y_2) + \frac{1}{2} P_{(\alpha\beta\gamma)}(Y_2) + \frac{1}{8} P_{(\alpha\beta\gamma\delta)}(Y_2)\\
&+P_{(\alpha\beta)}(Y_1) \Bigl(-\frac{1}{12} + \frac{1}{3} P_{(\alpha)}(Y_2) + \frac{1}{2} P_{(\alpha\beta)}(Y_2)  + \frac{1}{3} P_{(\alpha\beta\gamma)}(Y_2) + \frac{1}{12} P_{(\alpha\beta\gamma\delta)}(Y_2) \Bigr) \\
&+P_{(\alpha\beta)}(Y_2) \Bigl(-\frac{1}{12} + \frac{1}{3} P_{(\alpha)}(Y_1) + \frac{1}{2} P_{(\alpha\beta)}(Y_1) + \frac{1}{3} P_{(\alpha\beta\gamma)}(Y_1) + \frac{1}{12} P_{(\alpha\beta\gamma\delta)}(Y_1) \Bigr) \\
&+P_{(\alpha\beta\gamma)}(Y_1) \Bigl(-\frac{1}{4} + P_{(\alpha)}(Y_2) + \frac{3}{2} P_{(\alpha\beta)}(Y_2) + P_{(\alpha\beta\gamma)}(Y_2) + \frac{1}{4} P_{(\alpha\beta\gamma\delta)}(Y_2) \Bigr) \\
&+P_{(\alpha\beta\gamma)}(Y_2) \Bigl(-\frac{1}{4} + P_{(\alpha)}(Y_1) + \frac{3}{2} P_{(\alpha\beta)}(Y_1) + P_{(\alpha\beta\gamma)}(Y_1) + \frac{1}{4} P_{(\alpha\beta\gamma\delta)}(Y_1) \Bigr) \\
&+ P_{(\alpha\beta\gamma\delta)}(Y_1) \Bigl(-\frac{1}{8} + \frac{1}{2} P_{(\alpha)}(Y_2) + \frac{3}{4} P_{(\alpha\beta)}(Y_2) + \frac{1}{2} P_{(\alpha\beta\gamma)}(Y_2) + \frac{1}{8} P_{(\alpha\beta\gamma\delta)}(Y_2) \Bigr) \\
&+ P_{(\alpha\beta\gamma\delta)}(Y_2) \Bigl(-\frac{1}{8} + \frac{1}{2} P_{(\alpha)}(Y_1) + \frac{3}{4} P_{(\alpha\beta)}(Y_1) + \frac{1}{2} P_{(\alpha\beta\gamma)}(Y_1) + \frac{1}{8} P_{(\alpha\beta\gamma\delta)}(Y_1) \Bigr)
\end{align*}
\begin{align*}
= & \Bigl( \frac{1}{2}P_{(\alpha)}(Y_1)+ \frac{5}{4}P_{(\alpha\beta)}(Y_1)+ \frac{3}{4}P_{(\alpha\beta\gamma)}(Y_1) \Bigr) \Bigl( P_{(\alpha\beta)}(Y_2)-P_{(\beta\gamma)}(Y_2) \Bigr)\\
&+ \Bigl( \frac{1}{2}P_{(\alpha)}(Y_2)+ \frac{5}{4}P_{(\alpha\beta)}(Y_2)+ \frac{3}{4}P_{(\alpha\beta\gamma)}(Y_2) \Bigr) \Bigl( P_{(\alpha\beta)}(Y_1)-P_{(\beta\gamma)}(Y_1) \Bigr)\\
&+ \Bigl( \frac{1}{4}P_{(\alpha)}(Y_1)+ \frac{2}{3}P_{(\alpha\beta)}(Y_1)+ \frac{1}{4}P_{(\alpha\beta\gamma)}(Y_1) \Bigr) \Bigl( P_{(\alpha\beta\gamma)}(Y_2)-P_{(\beta\gamma\delta)}(Y_2) \Bigr)\\
&+ \Bigl( \frac{1}{4}P_{(\alpha)}(Y_2)+ \frac{2}{3}P_{(\alpha\beta)}(Y_2)+ \frac{1}{4}P_{(\alpha\beta\gamma)}(Y_2) \Bigr) \Bigl( P_{(\alpha\beta\gamma)}(Y_1)-P_{(\beta\gamma\delta)}(Y_1) \Bigr)\\
&+\frac{1}{8} \Bigl( 4P_{(\alpha)}(Y_1) + 6 P_{(\alpha\beta)}(Y_1) + 4 P_{(\alpha\beta\gamma)}(Y_1) + P_{(\alpha\beta\gamma\delta)}(Y_1) \Bigr)\\
&+\frac{1}{8} \Bigl( 4P_{(\alpha)}(Y_2) + 6 P_{(\alpha\beta)}(Y_2) + 4 P_{(\alpha\beta\gamma)}(Y_2) + P_{(\alpha\beta\gamma\delta)}(Y_2) \Bigr)\\
&+\frac{1}{12} P_{(\alpha\beta)}(Y_1) \Bigl( -1+ 4P_{(\alpha)}(Y_2) + 6 P_{(\alpha\beta)}(Y_2) + 4 P_{(\alpha\beta\gamma)}(Y_2) + P_{(\alpha\beta\gamma\delta)}(Y_2) \Bigr) \\
&+\frac{1}{12} P_{(\alpha\beta)}(Y_2) \Bigl( -1+ 4P_{(\alpha)}(Y_1) + 6 P_{(\alpha\beta)}(Y_1) + 4 P_{(\alpha\beta\gamma)}(Y_1) + P_{(\alpha\beta\gamma\delta)}(Y_1) \Bigr) \\
&+\frac{1}{4} P_{(\alpha\beta\gamma)}(Y_1) \Bigl( -1+ 4P_{(\alpha)}(Y_2) + 6 P_{(\alpha\beta)}(Y_2) + 4 P_{(\alpha\beta\gamma)}(Y_2) + P_{(\alpha\beta\gamma\delta)}(Y_2) \Bigr) \\
&+\frac{1}{4} P_{(\alpha\beta\gamma)}(Y_2) \Bigl(-1+ 4P_{(\alpha)}(Y_1) + 6 P_{(\alpha\beta)}(Y_1) + 4 P_{(\alpha\beta\gamma)}(Y_1) + P_{(\alpha\beta\gamma\delta)}(Y_1) \Bigr) \\
&+\frac{1}{8} P_{(\alpha\beta\gamma\delta)}(Y_1) \Bigl(-1+ 4P_{(\alpha)}(Y_2) + 6 P_{(\alpha\beta)}(Y_2) + 4 P_{(\alpha\beta\gamma)}(Y_2) + P_{(\alpha\beta\gamma\delta)}(Y_2) \Bigr) \\
&+\frac{1}{8} P_{(\alpha\beta\gamma\delta)}(Y_2) \Bigl(-1+ 4P_{(\alpha)}(Y_1) + 6 P_{(\alpha\beta)}(Y_1) + 4 P_{(\alpha\beta\gamma)}(Y_1) + P_{(\alpha\beta\gamma\delta)}(Y_1) \Bigr)
\end{align*}
\begin{align*}
= & \Bigl( \frac{1}{2}P_{(\alpha)}(Y_1)+ \frac{5}{4}P_{(\alpha\beta)}(Y_1)+ \frac{3}{4}P_{(\alpha\beta\gamma)}(Y_1) \Bigr) \Bigl( P_{(\alpha\beta)}(Y_2)-P_{(\beta\gamma)}(Y_2) \Bigr)\\
&+ \Bigl( \frac{1}{2}P_{(\alpha)}(Y_2)+ \frac{5}{4}P_{(\alpha\beta)}(Y_2)+ \frac{3}{4}P_{(\alpha\beta\gamma)}(Y_2) \Bigr) \Bigl( P_{(\alpha\beta)}(Y_1)-P_{(\beta\gamma)}(Y_1) \Bigr)\\
&+ \Bigl( \frac{1}{4}P_{(\alpha)}(Y_1)+ \frac{2}{3}P_{(\alpha\beta)}(Y_1)+ \frac{1}{4}P_{(\alpha\beta\gamma)}(Y_1) \Bigr) \Bigl( P_{(\alpha\beta\gamma)}(Y_2)-P_{(\beta\gamma\delta)}(Y_2) \Bigr)\\
&+ \Bigl( \frac{1}{4}P_{(\alpha)}(Y_2)+ \frac{2}{3}P_{(\alpha\beta)}(Y_2)+ \frac{1}{4}P_{(\alpha\beta\gamma)}(Y_2) \Bigr) \Bigl( P_{(\alpha\beta\gamma)}(Y_1)-P_{(\beta\gamma\delta)}(Y_1) \Bigr)\\
&+\frac{1}{2} \Bigl( p_1 + (1-4p_1) RA(Y_1) \Bigr) +\frac{1}{2} \Bigl( p_2 + (1-4p_2) RA(Y_2) \Bigr)\\
&+\frac{1}{12} P_{(\alpha\beta)}(Y_1) \Bigl( -1+ 4 \Bigl( p_2 + (1-4p_2) RA(Y_2) \Bigr) \Bigr) \\
&+\frac{1}{12} P_{(\alpha\beta)}(Y_2) \Bigl( -1+ 4 \Bigl( p_1 + (1-4p_1) RA(Y_1) \Bigr) \Bigr) \\
&+\frac{1}{4} P_{(\alpha\beta\gamma)}(Y_1) \Bigl( -1+ 4 \Bigl( p_2 + (1-4p_2) RA(Y_2) \Bigr) \Bigr) \\
&+\frac{1}{4} P_{(\alpha\beta\gamma)}(Y_2) \Bigl( -1+ 4 \Bigl( p_1 + (1-4p_1) RA(Y_1) \Bigr) \Bigr) \\
&+\frac{1}{8} P_{(\alpha\beta\gamma\delta)}(Y_1) \Bigl( -1+ 4 \Bigl( p_2 + (1-4p_2) RA(Y_2) \Bigr) \Bigr) \\
&+\frac{1}{8} P_{(\alpha\beta\gamma\delta)}(Y_2) \Bigl( -1+ 4 \Bigl( p_1 + (1-4p_1) RA(Y_1) \Bigr) \Bigr) \\
&\text{by } \eqref{piRA}
\end{align*}
\begin{align*}
= & \Bigl( \frac{1}{2}P_{(\alpha)}(Y_1)+ \frac{5}{4}P_{(\alpha\beta)}(Y_1)+ \frac{3}{4}P_{(\alpha\beta\gamma)}(Y_1) \Bigr) \Bigl( P_{(\alpha\beta)}(Y_2)-P_{(\beta\gamma)}(Y_2) \Bigr)\\
&+ \Bigl( \frac{1}{2}P_{(\alpha)}(Y_2)+ \frac{5}{4}P_{(\alpha\beta)}(Y_2)+ \frac{3}{4}P_{(\alpha\beta\gamma)}(Y_2) \Bigr) \Bigl( P_{(\alpha\beta)}(Y_1)-P_{(\beta\gamma)}(Y_1) \Bigr)\\
&+ \Bigl( \frac{1}{4}P_{(\alpha)}(Y_1)+ \frac{2}{3}P_{(\alpha\beta)}(Y_1)+ \frac{1}{4}P_{(\alpha\beta\gamma)}(Y_1) \Bigr) \Bigl( P_{(\alpha\beta\gamma)}(Y_2)-P_{(\beta\gamma\delta)}(Y_2) \Bigr)\\
&+ \Bigl( \frac{1}{4}P_{(\alpha)}(Y_2)+ \frac{2}{3}P_{(\alpha\beta)}(Y_2)+ \frac{1}{4}P_{(\alpha\beta\gamma)}(Y_2) \Bigr) \Bigl( P_{(\alpha\beta\gamma)}(Y_1)-P_{(\beta\gamma\delta)}(Y_1) \Bigr)\\
&+\frac{1}{2} \Bigl( p_1 + P_1 RA(Y_1) \Bigr) +\frac{1}{2} \Bigl( p_2 + P_2 RA(Y_2) \Bigr)\\
&+ \Bigl( \frac{1}{12} P_{(\alpha\beta)}(Y_1)+\frac{1}{4} P_{(\alpha\beta\gamma)}(Y_1) +\frac{1}{8} P_{(\alpha\beta\gamma\delta)}(Y_1) \Bigr) \Bigl( -1+ 4 \Bigl( p_2 + P_2 RA(Y_2) \Bigr) \Bigr) \\
&+\Bigl( \frac{1}{12} P_{(\alpha\beta)}(Y_2)+\frac{1}{4} P_{(\alpha\beta\gamma)}(Y_2)+\frac{1}{8} P_{(\alpha\beta\gamma\delta)}(Y_2) \Bigr) \Bigl( -1+ 4 \Bigl( p_1 + P_1 RA(Y_1) \Bigr) \Bigr) \\
&\text{by the definition of } P_1 \text{ and } P_2
\end{align*}Then
\begin{align*}
8D(X) =& 8RA(X) - 8 + 24p \\
= & \Bigl( 4P_{(\alpha)}(Y_1)+ 10P_{(\alpha\beta)}(Y_1)+ 6P_{(\alpha\beta\gamma)}(Y_1) \Bigr) \Bigl( P_{(\alpha\beta)}(Y_2)-P_{(\beta\gamma)}(Y_2) \Bigr)\\
&+ \Bigl( 4P_{(\alpha)}(Y_2)+ 10P_{(\alpha\beta)}(Y_2)+ 6P_{(\alpha\beta\gamma)}(Y_2) \Bigr) \Bigl( P_{(\alpha\beta)}(Y_1)-P_{(\beta\gamma)}(Y_1) \Bigr)\\
&+ \Bigl( 2P_{(\alpha)}(Y_1)+ \frac{16}{3}P_{(\alpha\beta)}(Y_1)+ 2P_{(\alpha\beta\gamma)}(Y_1) \Bigr) \Bigl( P_{(\alpha\beta\gamma)}(Y_2)-P_{(\beta\gamma\delta)}(Y_2) \Bigr)\\
&+ \Bigl( 2P_{(\alpha)}(Y_2)+ \frac{16}{3}P_{(\alpha\beta)}(Y_2)+ 2P_{(\alpha\beta\gamma)}(Y_2) \Bigr) \Bigl( P_{(\alpha\beta\gamma)}(Y_1)-P_{(\beta\gamma\delta)}(Y_1) \Bigr)\\
&+ 4 p_1 + 4 P_1 RA(Y_1) -4 +12p  +4 p_2 + 4P_2 RA(Y_2) -4 +12p \\
&+ \Bigl( \frac{2}{3} P_{(\alpha\beta)}(Y_1)+2 P_{(\alpha\beta\gamma)}(Y_1) + P_{(\alpha\beta\gamma\delta)}(Y_1) \Bigr) \Bigl( -1+ 4 p_2 + 4P_2 RA(Y_2) \Bigr) \\
&+\Bigl( \frac{2}{3} P_{(\alpha\beta)}(Y_2)+2 P_{(\alpha\beta\gamma)}(Y_2)+ P_{(\alpha\beta\gamma\delta)}(Y_2) \Bigr) \Bigl( -1+ 4 p_1 + 4P_1 RA(Y_1) \Bigr)
\end{align*}

\begin{align*}
= & \Bigl( 4P_{(\alpha)}(Y_1)+ 10P_{(\alpha\beta)}(Y_1)+ 6P_{(\alpha\beta\gamma)}(Y_1) \Bigr) \Bigl( P_{(\alpha\beta)}(Y_2)-P_{(\beta\gamma)}(Y_2) \Bigr)\\
&+ \Bigl( 4P_{(\alpha)}(Y_2)+ 10P_{(\alpha\beta)}(Y_2)+ 6P_{(\alpha\beta\gamma)}(Y_2) \Bigr) \Bigl( P_{(\alpha\beta)}(Y_1)-P_{(\beta\gamma)}(Y_1) \Bigr)\\
&+ \Bigl( 2P_{(\alpha)}(Y_1)+ \frac{16}{3}P_{(\alpha\beta)}(Y_1)+ 2P_{(\alpha\beta\gamma)}(Y_1) \Bigr) \Bigl( P_{(\alpha\beta\gamma)}(Y_2)-P_{(\beta\gamma\delta)}(Y_2) \Bigr)\\
&+ \Bigl( 2P_{(\alpha)}(Y_2)+ \frac{16}{3}P_{(\alpha\beta)}(Y_2)+ 2P_{(\alpha\beta\gamma)}(Y_2) \Bigr) \Bigl( P_{(\alpha\beta\gamma)}(Y_1)-P_{(\beta\gamma\delta)}(Y_1) \Bigr)\\
&+4 P_1 RA(Y_1) +4P_1 (-1 + 3p_1^{'}) + 4P_2 RA(Y_2) +4P_2 (-1 + 3p_2^{'}) \\
&+ \Bigl( \frac{2}{3} P_{(\alpha\beta)}(Y_1)+2 P_{(\alpha\beta\gamma)}(Y_1) + P_{(\alpha\beta\gamma\delta)}(Y_1) \Bigr) \Bigl( 3P + 4P_2 (-1 + 3p_2^{'}) + 4P_2 RA(Y_2) \Bigr) \\
&+\Bigl( \frac{2}{3} P_{(\alpha\beta)}(Y_2)+2 P_{(\alpha\beta\gamma)}(Y_2)+ P_{(\alpha\beta\gamma\delta)}(Y_2) \Bigr) \Bigl( 3P + 4P_1 (-1 + 3p_1^{'}) + 4P_1 RA(Y_1) \Bigr) \\
&\text{by } \eqref{4pi1}, \eqref{4pi2}
\end{align*}

\begin{align*}
= & \Bigl( 4P_{(\alpha)}(Y_1)+ 10P_{(\alpha\beta)}(Y_1)+ 6P_{(\alpha\beta\gamma)}(Y_1) \Bigr) \Bigl( P_{(\alpha\beta)}(Y_2)-P_{(\beta\gamma)}(Y_2) \Bigr)\\
&+ \Bigl( 4P_{(\alpha)}(Y_2)+ 10P_{(\alpha\beta)}(Y_2)+ 6P_{(\alpha\beta\gamma)}(Y_2) \Bigr) \Bigl( P_{(\alpha\beta)}(Y_1)-P_{(\beta\gamma)}(Y_1) \Bigr)\\
&+ \Bigl( 2P_{(\alpha)}(Y_1)+ \frac{16}{3}P_{(\alpha\beta)}(Y_1)+ 2P_{(\alpha\beta\gamma)}(Y_1) \Bigr) \Bigl( P_{(\alpha\beta\gamma)}(Y_2)-P_{(\beta\gamma\delta)}(Y_2) \Bigr)\\
&+ \Bigl( 2P_{(\alpha)}(Y_2)+ \frac{16}{3}P_{(\alpha\beta)}(Y_2)+ 2P_{(\alpha\beta\gamma)}(Y_2) \Bigr) \Bigl( P_{(\alpha\beta\gamma)}(Y_1)-P_{(\beta\gamma\delta)}(Y_1) \Bigr)\\
&+4 P_1 (RA(Y_1) -1 + 3p_1^{'}) + 4P_2 (RA(Y_2) -1 + 3p_2^{'}) \\
&+ \Bigl( \frac{2}{3} P_{(\alpha\beta)}(Y_1)+2 P_{(\alpha\beta\gamma)}(Y_1) + P_{(\alpha\beta\gamma\delta)}(Y_1) \Bigr) \Bigl( 3P + 4P_2 (-1 + 3p_2^{'} + RA(Y_2)) \Bigr) \\
&+\Bigl( \frac{2}{3} P_{(\alpha\beta)}(Y_2)+2 P_{(\alpha\beta\gamma)}(Y_2)+ P_{(\alpha\beta\gamma\delta)}(Y_2) \Bigr) \Bigl( 3P + 4P_1 (-1 + 3p_1^{'}+ RA(Y_1)) \Bigr)
\end{align*}
\begin{align*}
= & \Bigl( 4P_{(\alpha)}(Y_1)+ 10P_{(\alpha\beta)}(Y_1)+ 6P_{(\alpha\beta\gamma)}(Y_1) \Bigr) \Bigl( P_{(\alpha\beta)}(Y_2)-P_{(\beta\gamma)}(Y_2) \Bigr)\\
&+ \Bigl( 4P_{(\alpha)}(Y_2)+ 10P_{(\alpha\beta)}(Y_2)+ 6P_{(\alpha\beta\gamma)}(Y_2) \Bigr) \Bigl( P_{(\alpha\beta)}(Y_1)-P_{(\beta\gamma)}(Y_1) \Bigr)\\
&+ \Bigl( 2P_{(\alpha)}(Y_1)+ \frac{16}{3}P_{(\alpha\beta)}(Y_1)+ 2P_{(\alpha\beta\gamma)}(Y_1) \Bigr) \Bigl( P_{(\alpha\beta\gamma)}(Y_2)-P_{(\beta\gamma\delta)}(Y_2) \Bigr)\\
&+ \Bigl( 2P_{(\alpha)}(Y_2)+ \frac{16}{3}P_{(\alpha\beta)}(Y_2)+ 2P_{(\alpha\beta\gamma)}(Y_2) \Bigr) \Bigl( P_{(\alpha\beta\gamma)}(Y_1)-P_{(\beta\gamma\delta)}(Y_1) \Bigr)\\
&+4 P_1 D_1 + 4P_2 D_2 \\
&+ \Bigl( \frac{2}{3} P_{(\alpha\beta)}(Y_1)+2 P_{(\alpha\beta\gamma)}(Y_1) + P_{(\alpha\beta\gamma\delta)}(Y_1) \Bigr) \Bigl( 3P + 4P_2 D_2 \Bigr) \\
&+\Bigl( \frac{2}{3} P_{(\alpha\beta)}(Y_2)+2 P_{(\alpha\beta\gamma)}(Y_2)+ P_{(\alpha\beta\gamma\delta)}(Y_2) \Bigr) \Bigl( 3P + 4P_1 D_1 \Bigr) \\
&\text{by the definition of } D_1, D_2 
\end{align*}

\begin{align*}
= & \Bigl( 4P_{(\alpha)}(Y_1)+ 10P_{(\alpha\beta)}(Y_1)+ 6P_{(\alpha\beta\gamma)}(Y_1) \Bigr) \Bigl( P_{\alpha\beta}(Y_2)-P_{\beta\gamma}(Y_2) \Bigr)P_2 \\
&+ \Bigl( 4P_{(\alpha)}(Y_2)+ 10P_{(\alpha\beta)}(Y_2)+ 6P_{(\alpha\beta\gamma)}(Y_2) \Bigr) \Bigl( P_{\alpha\beta}(Y_1)-P_{\beta\gamma}(Y_1) \Bigr)P_1\\
&+ \Bigl( 2P_{(\alpha)}(Y_1)+ \frac{16}{3}P_{(\alpha\beta)}(Y_1)+ 2P_{(\alpha\beta\gamma)}(Y_1) \Bigr) \Bigl( P_{\alpha\beta\gamma}(Y_2)-P_{\beta\gamma\delta}(Y_2) \Bigr) P_2\\
&+ \Bigl( 2P_{(\alpha)}(Y_2)+ \frac{16}{3}P_{(\alpha\beta)}(Y_2)+ 2P_{(\alpha\beta\gamma)}(Y_2) \Bigr) \Bigl( P_{\alpha\beta\gamma}(Y_1)-P_{\beta\gamma\delta}(Y_1) \Bigr)P_1\\
&+4 P_1 D_1 + 4P_2 D_2 \\
&+ \Bigl( \frac{2}{3} P_{(\alpha\beta)}(Y_1)+2 P_{(\alpha\beta\gamma)}(Y_1) + P_{(\alpha\beta\gamma\delta)}(Y_1) \Bigr) \Bigl( 3P + 4P_2 D_2 \Bigr) \\
&+\Bigl( \frac{2}{3} P_{(\alpha\beta)}(Y_2)+2 P_{(\alpha\beta\gamma)}(Y_2)+ P_{(\alpha\beta\gamma\delta)}(Y_2) \Bigr) \Bigl( 3P + 4P_1 D_1 \Bigr) \\
&\text{by } \eqref{4(ab)(bc)}, \eqref{4(abc)(bcd)}
\end{align*}

\begin{lemma} \label{lemmageq3}
For any rooted binary phylogenetic tree and the $N_3$-model we have that
\begin{align*}
&P_{\alpha}(X) \geq P_{\beta}(X), \\
&P_{\alpha\beta}(X) \geq P_{\beta\gamma}(X).
\end{align*}
\end{lemma}
\noindent
Note that Lemma \ref{lemmageq3} does also not require the underlying tree to be ultrametric.
\begin{proof}
To prove Lemma \ref{lemmageq3} we show that for any rooted binary phylogenetic tree $T$ under a symmetric 3-state substitution model
\begin{align}
&P_{\alpha}(X) \geq P_{\beta}(X), \label{3lemma1} \\
&P_{\alpha\beta}(X) \geq P_{\beta\gamma}(X), \label{3lemma2}
\end{align} by induction on $n$. For $n=2$ the subtrees $Y_1$ and $Y_2$ both contain one leaf, and hence $p=p_1=p_2$ leads to
\begin{align*}
&P_{\alpha}(X)= (1-2p)^2, \\
&P_{\beta}(X)= p^2, \\
&P_{\alpha\beta}(X)= 2(1-2p)p, \\
&P_{\beta\gamma}(X)= 2p^2.
\end{align*} Therefore
\begin{align*}
P_{\alpha}(X) - P_{\beta}(X) &= (1-2p)^2 - p^2 = 1-4p+4p^2-p^2= 1-4p+3p^2 \\
&=\underbrace{(1-3p)}_{\geq 0} \underbrace{(1-p)}_{\geq 0} \geq 0 \text{ as } p \leq \frac{1}{3}.
\end{align*}
Moreover 
\begin{align*}
P_{\alpha\beta}(X) - P_{\beta\gamma}(X) &= 2(1-2p)p - 2p^2 = 2p \underbrace{(1-3p)}_{\geq 0} \geq 0 \text{ as } p \leq \frac{1}{3},
\end{align*}
which completes the base case of the induction. For the inductive step we first define some recursions similar to \eqref{w1}, \eqref{w2}, \eqref{w3}, \eqref{w4} and \eqref{w5}:
\begin{align}
&P_{(\alpha)}(Y_i) =(1-2p_i) P_{\alpha}(Y_i) + 2p_i P_{\beta}(Y_i), \label{w13} \\
&P_{(\beta)}(Y_i) = (1-p_i) P_{\beta}(Y_i) + p_i P_{\alpha}(Y_i) = P_{(\gamma)}(Y_i), \label{w23} \\
&P_{(\alpha\beta)}(Y_i) = (1-p_i) P_{\alpha\beta}(Y_i) + p_i P_{\beta\gamma}(Y_i)=P_{(\alpha\gamma)}(Y_i), \label{w33} \\
&P_{(\beta\gamma)}(Y_i) = (1-2p_i) P_{\beta\gamma}(Y_i) + 2p_i P_{\alpha\beta}(Y_i), \label{w43} \\
&P_{(\alpha\beta\gamma)}(Y_i) =P_{\alpha\beta\gamma}(Y_i), \label{w53} 
\end{align} With \eqref{w13}, \eqref{w23}, \eqref{w33}, \eqref{w43} and \eqref{w53} we therefore have:
\begin{align}
P_{\alpha}(X)= &P_{(\alpha)}(Y_1) P_{(\alpha)}(Y_2) + 2P_{(\alpha)}(Y_1) P_{(\alpha\beta)}(Y_2)  + 2P_{(\alpha\beta)}(Y_1) P_{(\alpha)}(Y_2)  \nonumber \\
&+ 2P_{(\alpha\beta)}(Y_1) P_{(\alpha\beta)}(Y_2) + P_{(\alpha)}(Y_1) P_{(\alpha\beta\gamma)}(Y_2) + P_{(\alpha\beta\gamma)}(Y_1) P_{(\alpha)}(Y_2)
\end{align}
\begin{align}
P_{\beta}(X)= &P_{(\beta)}(Y_1) P_{(\beta)}(Y_2) + P_{(\beta)}(Y_1) P_{(\alpha\beta)}(Y_2) + P_{(\alpha\beta)}(Y_1) P_{(\beta)}(Y_2) \nonumber \\
&+ P_{(\beta)}(Y_1) P_{(\beta\gamma)}(Y_2) + P_{(\beta\gamma)}(Y_1) P_{(\beta)}(Y_2) + P_{(\alpha\beta)}(Y_1) P_{(\beta\gamma)}(Y_2) \nonumber \\
&+ P_{(\beta\gamma)}(Y_1) P_{(\alpha\beta)}(Y_2) + P_{(\beta)}(Y_1) P_{(\alpha\beta\gamma)}(Y_2) + P_{(\alpha\beta\gamma)}(Y_1) P_{(\beta)}(Y_2)
\end{align}
\begin{align}
P_{\alpha\beta}(X) =&P_{(\alpha)}(Y_1) P_{(\beta)}(Y_2) + P_{(\beta)}(Y_1) P_{(\alpha)}(Y_2)+ P_{(\alpha\beta)}(Y_1) P_{(\alpha\beta)}(Y_2) \nonumber \\
&+ P_{(\alpha\beta)}(Y_1) P_{(\alpha\beta\gamma)}(Y_2) + P_{(\alpha\beta\gamma)}(Y_1) P_{(\alpha\beta)}(Y_2)
\end{align}
\begin{align}
P_{\beta\gamma}(X) =&2P_{(\beta)}(Y_1) P_{(\beta)}(Y_2) + P_{(\beta\gamma)}(Y_1) P_{(\beta\gamma)}(Y_2) \nonumber \\
&+ P_{(\beta\gamma)}(Y_1) P_{(\alpha\beta\gamma)}(Y_2) + P_{(\alpha\beta\gamma)}(Y_1) P_{(\beta\gamma)}(Y_2)
\end{align}
Moreover we have that for $i \in \{1,2\}$
\begin{align}
P_{(\alpha)}(Y_i)-P_{(\beta)}(Y_i)
&= (1-2p_i) P_{\alpha}(Y_i) + 2p_i P_{\beta}(Y_i) - (1-p_i) P_{\beta}(Y_i) - p_i P_{\alpha}(Y_i) \nonumber \\
&\qquad \text{by } \eqref{w13}, \eqref{w23} \nonumber \\
&=(1-3p_i) P_{\alpha}(Y_i) - (1-3p_i) P_{\beta}(Y_i) \nonumber \\
&=(1-3p_i) \Bigl( P_{\alpha}(Y_i)-P_{\beta}(Y_i) \Bigr) \label{4(a)(b)3states}
\end{align} and thus
\begin{align}
P_{(\alpha)}(Y_i) = P_{(\beta)}(Y_i) + (1-3p_i) \Bigl( P_{\alpha}(Y_i)-P_{\beta}(Y_i) \Bigr). \label{4(a)3states}
\end{align}
In the same manner by \eqref{w33} and \eqref{w43} we can see that
\begin{align}
&P_{(\alpha\beta)}(Y_i)-P_{(\beta\gamma)}(Y_i)= (1-3p_i) \Bigl( P_{\alpha\beta}(Y_i)-P_{\beta\gamma}(Y_i) \Bigr), \label{4(ab)(bc)3states}.
\end{align} Therefore
\begin{align}
&P_{(\alpha\beta)}(Y_i) = P_{(\beta\gamma)}(Y_i) + (1-3p_i) \Bigl( P_{\alpha\beta}(Y_i)-P_{\beta\gamma}(Y_i) \Bigr). \label{4(ab)3states}
\end{align} Additionally we have the following: choose sets $A_1, A_2$ from $\{\{\alpha\},\{\alpha\beta\}\}$ and $B_1, B_2$ from $\{\{\beta\},\{\beta\gamma\}\}$ such that for $i \in \{1,2\}$ $|A_i|=|B_i|$, respectively. Then we have that
\begin{align}
&P_{(A_1)}(Y_1)P_{(A_2)}(Y_2)-P_{(B_1)}(Y_1)P_{(B_2)}(Y_2) \nonumber \nonumber \\
= &\Bigl( P_{(B_1)}(Y_1) + (1-3p_1) \Bigl( P_{A_1}(Y_1)-P_{B_1}(Y_1) \Bigr) \Bigr)
\Bigl( P_{(B_2)}(Y_2) + (1-3p_2) \Bigl( P_{A_2}(Y_2)-P_{B_2}(Y_2) \Bigr) \Bigr) \nonumber \\
&- P_{(B_1)}(Y_1)P_{(B_2)}(Y_2) \nonumber \\
&\text{by } \eqref{4(a)3states} \text{ or } \eqref{4(ab)3states} \nonumber \\
= &P_{(B_1)}(Y_1)P_{(B_2)}(Y_2) \nonumber \\
&+ P_{(B_1)}(Y_1)(1-3p_2) \Bigl( P_{A_2}(Y_2)-P_{B_2}(Y_2) \Bigr) + P_{(B_2)}(Y_2) (1-3p_1) \Bigl( P_{A_1}(Y_1)-P_{B_1}(Y_1) \Bigr) \nonumber \\
&+ (1-3p_1)(1-3p_2) \Bigl( P_{A_1}(Y_1)-P_{B_1}(Y_1) \Bigr) \Bigl( P_{A_2}(Y_2)-P_{B_2}(Y_2) \Bigr) - P_{(B_1)}(Y_1)P_{(B_2)}(Y_2) \nonumber \\
= &P_{(B_1)}(Y_1)(1-3p_2) \Bigl( P_{A_2}(Y_2)-P_{B_2}(Y_2) \Bigr) + P_{(B_2)}(Y_2) (1-3p_1) \Bigl( P_{A_1}(Y_1)-P_{B_1}(Y_1) \Bigr) \nonumber \\
&+ (1-3p_1)(1-3p_2) \Bigl( P_{A_1}(Y_1)-P_{B_1}(Y_1) \Bigr) \Bigl( P_{A_2}(Y_2)-P_{B_2}(Y_2) \Bigr) . \label{3states4(a)(a)(b)(b)3states}
\end{align} Now suppose that $T$ has $n$ taxa and that \eqref{3lemma1} and \eqref{3lemma2} are true for all trees having fewer that $n$ taxa. Note that therefore \eqref{3states4(a)(a)(b)(b)3states} is non-negative, since $Y_1$ and $Y_2$ contain both fewer than than $n$ taxa. Then
\begin{align*}
&P_{\alpha}(X) - P_{\beta}(X) \\
=&P_{(\alpha)}(Y_1) P_{(\alpha)}(Y_2) + 2P_{(\alpha)}(Y_1) P_{(\alpha\beta)}(Y_2)  + 2P_{(\alpha\beta)}(Y_1) P_{(\alpha)}(Y_2)  \\
&+ 2P_{(\alpha\beta)}(Y_1) P_{(\alpha\beta)}(Y_2) + P_{(\alpha)}(Y_1) P_{(\alpha\beta\gamma)}(Y_2) + P_{(\alpha\beta\gamma)}(Y_1) P_{(\alpha)}(Y_2) \\
&- P_{(\beta)}(Y_1) P_{(\beta)}(Y_2) - P_{(\beta)}(Y_1) P_{(\alpha\beta)}(Y_2) - P_{(\alpha\beta)}(Y_1) P_{(\beta)}(Y_2) \\
&- P_{(\beta)}(Y_1) P_{(\beta\gamma)}(Y_2) - P_{(\beta\gamma)}(Y_1) P_{(\beta)}(Y_2) - P_{(\alpha\beta)}(Y_1) \cdot P_{(\beta\gamma)}(Y_2) \\
&- P_{(\beta\gamma)}(Y_1) P_{(\alpha\beta)}(Y_2) - P_{(\beta)}(Y_1) P_{(\alpha\beta\gamma)}(Y_2) - P_{(\alpha\beta\gamma)}(Y_1) P_{(\beta)}(Y_2) \\
=&P_{(\alpha)}(Y_1) P_{(\alpha)}(Y_2) - P_{(\beta)}(Y_1) P_{(\beta)}(Y_2) \\
&+ P_{(\alpha\beta)}(Y_2) \Bigl( P_{(\alpha)}(Y_1) - P_{(\beta)}(Y_1) \Bigr) + P_{(\alpha\beta)}(Y_1) \Bigl( P_{(\alpha)}(Y_2) - P_{(\beta)}(Y_2) \Bigr) \\
&+ P_{(\alpha)}(Y_1) P_{(\alpha\beta)}(Y_2) - P_{(\beta)}(Y_1) P_{(\beta\gamma)}(Y_2) \\
&+ P_{(\alpha\beta)}(Y_1) P_{(\alpha)}(Y_2) - P_{(\beta\gamma)}(Y_1) P_{(\beta)}(Y_2) \\
&+ P_{(\alpha\beta)}(Y_2) \Bigl( P_{(\alpha\beta)}(Y_1) - P_{(\beta\gamma)}(Y_1) \Bigr) + P_{(\alpha\beta)}(Y_1) \Bigl( P_{(\alpha\beta)}(Y_2) - P_{(\beta\gamma)}(Y_2) \Bigr) \\
&+ P_{(\alpha\beta\gamma)}(Y_2) \Bigl( P_{(\alpha)}(Y_1) - P_{(\beta)}(Y_1) \Bigr) + P_{(\alpha\beta\gamma)}(Y_1) \Bigl( P_{(\alpha)}(Y_2) - P_{(\beta)}(Y_2) \Bigr) \\
= &P_{(\beta)}(Y_1)(1-3p_2) \Bigl( P_{\alpha}(Y_2)-P_{\beta}(Y_2) \Bigr) + P_{(\beta)}(Y_2) (1-3p_1) \Bigl( P_{\alpha}(Y_1)-P_{\beta}(Y_1) \Bigr)  \nonumber \\
&+ (1-3p_1)(1-3p_2) \Bigl( P_{\alpha}(Y_1)-P_{\beta}(Y_1) \Bigr) \Bigl( P_{\alpha}(Y_2)-P_{\beta}(Y_2) \Bigr)\\
&+ P_{(\alpha\beta)}(Y_2) (1-3p_1) \Bigl( P_{\alpha}(Y_1)-P_{\beta}(Y_1) \Bigr) + P_{(\alpha\beta)}(Y_1) (1-3p_2) \Bigl( P_{\alpha}(Y_2)-P_{\beta}(Y_2) \Bigr) \\
&+ P_{(\beta)}(Y_1)(1-3p_2) \Bigl( P_{\alpha\beta}(Y_2)-P_{\beta\gamma}(Y_2) \Bigr) + P_{(\beta\gamma)}(Y_2) (1-3p_1) \Bigl( P_{\alpha}(Y_1)-P_{\beta}(Y_1) \Bigr)  \nonumber \\
&+ (1-3p_1)(1-3p_2) \Bigl( P_{\alpha}(Y_1)-P_{\beta}(Y_1) \Bigr) \Bigl( P_{\alpha\beta}(Y_2)-P_{\beta\gamma}(Y_2) \Bigr)\\
&+ P_{(\beta\gamma)}(Y_1)(1-3p_2) \Bigl( P_{\alpha}(Y_2)-P_{\beta}(Y_2) \Bigr) + P_{(\beta)}(Y_2) (1-3p_1) \Bigl( P_{\alpha\beta}(Y_1)-P_{\beta\gamma}(Y_1) \Bigr)  \nonumber \\
&+ (1-3p_1)(1-3p_2) \Bigl( P_{\alpha\beta}(Y_1)-P_{\beta\gamma}(Y_1) \Bigr) \Bigl( P_{\alpha}(Y_2)-P_{\beta}(Y_2) \Bigr) \\
&+ P_{(\alpha\beta)}(Y_2) (1-3p_1) \Bigl( P_{\alpha\beta}(Y_1)-P_{\beta\gamma}(Y_1) \Bigr) + P_{(\alpha\beta)}(Y_1) (1-3p_2) \Bigl( P_{\alpha\beta}(Y_2)-P_{\beta\gamma}(Y_2) \Bigr) \\
&+ P_{(\alpha\beta\gamma)}(Y_2) (1-3p_1) \Bigl( P_{\alpha}(Y_1)-P_{\beta}(Y_1) \Bigr) + P_{(\alpha\beta\gamma)}(Y_1) (1-3p_2) \Bigl( P_{\alpha}(Y_2)-P_{\beta}(Y_2) \Bigr) \\
&\text{by } \eqref{4(a)(b)3states}, \eqref{4(ab)(bc)3states}, \eqref{3states4(a)(a)(b)(b)3states}
\end{align*}
By the inductive assumption this term is non-negative, and therefore concludes the proof for $P_{\alpha}(X) \geq P_{\beta}(X)$. We now proceed with the second part of Lemma \ref{lemmageq3}.
\begin{align*}
&P_{\alpha\beta}(X) - P_{\beta\gamma}(X) \\
=&P_{(\alpha)}(Y_1) P_{(\beta)}(Y_2) + P_{(\beta)}(Y_1) P_{(\alpha)}(Y_2)+ P_{(\alpha\beta)}(Y_1) P_{(\alpha\beta)}(Y_2) \\
&+ P_{(\alpha\beta)}(Y_1) P_{(\alpha\beta\gamma)}(Y_2) + P_{(\alpha\beta\gamma)}(Y_1) P_{(\alpha\beta)}(Y_2) - 2P_{(\beta)}(Y_1) P_{(\beta)}(Y_2) \\
&- P_{(\beta\gamma)}(Y_1) P_{(\beta\gamma)}(Y_2) - P_{(\beta\gamma)}(Y_1) P_{(\alpha\beta\gamma)}(Y_2) - P_{(\alpha\beta\gamma)}(Y_1) P_{(\beta\gamma)}(Y_2) \\
= &P_{(\beta)}(Y_2) \Bigl( P_{(\alpha)}(Y_1) - P_{(\beta)}(Y_1) \Bigr) + P_{(\beta)}(Y_1) \Bigl( P_{(\alpha)}(Y_2) - P_{(\beta)}(Y_2) \Bigr) \\
&+ P_{(\alpha\beta)}(Y_1) P_{(\alpha\beta)}(Y_2) - P_{(\beta\gamma)}(Y_1) P_{(\beta\gamma)}(Y_2) \\
&+ P_{(\alpha\beta\gamma)}(Y_2) \Bigl( P_{(\alpha\beta)}(Y_1) - P_{(\beta\gamma)}(Y_1) \Bigr) + P_{(\alpha\beta\gamma)}(Y_1) \Bigl( P_{(\alpha\beta)}(Y_2) - P_{(\beta\gamma)}(Y_2) \Bigr) \\
= &P_{(\beta)}(Y_2) (1-3p_1) \Bigl( P_{\alpha}(Y_1) - P_{\beta}(Y_1) \Bigr) + P_{(\beta)}(Y_1) (1-3p_2) \Bigl( P_{\alpha}(Y_2) - P_{\beta}(Y_2) \Bigr) \\
&+ P_{(\beta\gamma)}(Y_1)(1-3p_2) \Bigl( P_{\alpha\beta}(Y_2)-P_{\beta\gamma}(Y_2) \Bigr) + P_{(\beta\gamma)}(Y_2) (1-3p_1) \Bigl( P_{\alpha\beta}(Y_1)-P_{\beta\gamma}(Y_1) \Bigr)  \nonumber \\
&+ (1-3p_1)(1-3p_2) \Bigl( P_{\alpha\beta}(Y_1)-P_{\beta\gamma}(Y_1) \Bigr) \Bigl( P_{\alpha\beta}(Y_2)-P_{\beta\gamma}(Y_2) \Bigr)\\
&+ P_{(\alpha\beta\gamma)}(Y_2) (1-3p_1) \Bigl( P_{\alpha\beta}(Y_1) - P_{\beta\gamma}(Y_1) \Bigr) + P_{(\alpha\beta\gamma)}(Y_1) (1-3p_2) \Bigl( P_{\alpha\beta}(Y_2) - P_{\beta\gamma}(Y_2) \Bigr)\\
&\text{by } \eqref{4(a)(b)3states}, \eqref{4(ab)(bc)3states}, \eqref{3states4(a)(a)(b)(b)3states}
\end{align*}
By inductive assumption $P_{\alpha\beta}(X) - P_{\beta\gamma}(X)$ is non-negative, and therefore concludes the proof of the second part of Lemma \ref{lemmageq3}.
\end{proof}
\begin{corollary}\label{twoleaves3states}
Let $T$ be a rooted binary ultrametric phylogenetic tree on taxon set $X$ with $|X|=2$. Let $p$ denote the probability of change from the root to any leaf under the $N_3$-model. Then,
the reconstruction accuracy for ancestral state reconstruction using the Fitch algorithm is given by
$$RA(X)=1-2p.$$
\end{corollary}
\begin{proposition}
For any rooted binary phylogenetic ultrametric tree and the $N_3$-model, the reconstruction accuracy for the Fitch algorithm using any two terminal taxa $x_1,x_2 \in X$ for ancestral state reconstruction is given by
$$RA(\{x_1,x_2\}) = 1-2p.$$
\end{proposition}
\begin{proof}
Let $x_1, x_2 \in X$ be two terminal taxa of any rooted binary ultrametric phylogenetic tree $T$. Moreover, we consider the standard decomposition of $T$ into its two maximal pending subtrees $T_1$ and $T_2$ as depicted in Figure \ref{tree}. Thus, the proof is divided  into two cases. \\
In the first case we have without loss of generality $x_1 \in Y_1$ and $x_2 \in Y_2$. By Corollary \ref{twoleaves3states} the reconstruction accuracy using $x_1$ and $x_2$ is then $RA(\{x_1,x_2\})=1-2p$. 
\\
In the second case we have either $x_1, x_2 \in Y_1$ or $x_1,x_2 \in Y_2$. Thus, without loss of generality we consider $x_1, x_2 \in Y_1$ as depicted in Figure \ref{tree2}. Let $y$ be the last common ancestor of $x_1$ and $x_2$, i.e. the first node that occurs both on the path from $x_1$ to $\rho$ as well as on the path from $x_2$ to $\rho$. Let $\widehat{T}$ be the subtree of $T_1$ consisting of the paths from $y$ to $x_1$ and $x_2$, respectively. $\widehat{T}$ is depicted with dotted lines in Figure \ref{tree2}. Thus, the root of $\widehat{T}$ is $y$. In addition, let $\overline{p}$ be the probability for one specific change from $\rho$ to $y$, and let $\widehat{p}$ be the probability for one specific change from $y$ to $x_1$ or $x_2$.
By \eqref{ra3states} we have 
\begin{align}
RA(\{x_1,x_2\})=P_{\alpha}(\{x_1,x_2\}) + P_{\alpha\beta}(\{x_1,x_2\}). \label{ra2taxa3states}
\end{align}
Note that $P_{\alpha\beta\gamma}(\{x_1,x_2\})=0$ since we cannot obtain sets with more than two elements with the Fitch algorithm when only $x_1$ and $x_2$ are used for the reconstruction.
\\
In the following, we use the notation $f|_{\{x_1,x_2\}}$ for the restriction of character $f$ on taxa $x_1$ and $x_2$. \\
Furthermore, we have
\begin{align}
P_{\alpha}(\{x_1,x_2\}) &= \PP(\MP(f|_{\{x_1,x_2\}},\widehat{T})=\{\alpha\} | \rho=\alpha) \nonumber \\
&= (1-2\overline{p})~ \PP(\MP(f|_{\{x_1,x_2\}},\widehat{T})=\{\alpha\} | y=\alpha,\rho=\alpha) \nonumber \\
&\quad + \overline{p}~ \PP(\MP(f|_{\{x_1,x_2\}},\widehat{T})=\{\alpha\} | y=\beta,\rho=\alpha) \nonumber \\
&\quad + \overline{p}~ \PP(\MP(f|_{\{x_1,x_2\}},\widehat{T})=\{\alpha\} | y=\gamma,\rho=\alpha) \nonumber \\
&= (1-2\overline{p})~ \PP(\MP(f|_{\{x_1,x_2\}},\widehat{T})=\{\alpha\} | y=\alpha,\rho=\alpha) \nonumber \\
&\quad + 2\overline{p}~ \PP(\MP(f|_{\{x_1,x_2\}},\widehat{T})=\{\beta\} | y=\alpha,\rho=\alpha) \nonumber \\
&\qquad \text{by the symmetry of the } N_3 \text{-model} \nonumber \\
&= (1-2\overline{p}) (1-2\widehat{p})^2 + 2\overline{p} \widehat{p}^2 \label{2alpha3states}
\end{align}
Moreover,
\begin{align}
P_{\alpha\beta}(\{x_1,x_2\}) &= \PP(\MP(f|_{\{x_1,x_2\}},\widehat{T})=\{\alpha,\beta\} | \rho=\alpha) \nonumber \\
&= (1-2\overline{p})~ \PP(\MP(f|_{\{x_1,x_2\}},\widehat{T})=\{\alpha,\beta\} | y=\alpha,\rho=\alpha) \nonumber \\
&\quad + \overline{p}~ \PP(\MP(f|_{\{x_1,x_2\}},\widehat{T})=\{\alpha,\beta\} | y=\beta,\rho=\alpha) \nonumber \\
&\quad + \overline{p}~ \PP(\MP(f|_{\{x_1,x_2\}},\widehat{T})=\{\alpha,\beta\} | y=\gamma,\rho=\alpha) \nonumber \\
&= (1-\overline{p})~ \PP(\MP(f|_{\{x_1,x_2\}},\widehat{T})=\{\alpha,\beta\} | y=\alpha,\rho=\alpha) \nonumber \\
&\quad + \overline{p}~ \PP(\MP(f|_{\{x_1,x_2\}},\widehat{T})=\{\beta,\gamma\} | y=\alpha,\rho=\alpha) \nonumber \\
&\qquad \text{by the symmetry of the } N_3 \text{-model} \nonumber \\
&= (1-\overline{p}) ~2~ (1-2\widehat{p})~ \widehat{p} + \overline{p} ~2~ \widehat{p}^2 \label{2alphabeta3states}
\end{align}
Thus by \ref{2alpha3states} and \eqref{2alphabeta3states}, \eqref{ra2taxa3states} becomes
\begin{align*}
RA(\{x_1,x_2\})&= (1-2\overline{p}) (1-2\widehat{p})^2 + 2\overline{p} \widehat{p}^2 + (1-\overline{p}) ~2~ (1-2\widehat{p})~ \widehat{p} + \overline{p} ~2~ \widehat{p}^2 \\
&=1- 2~\overline{p} -2~\widehat{p} + 6~\overline{p}~\widehat{p} \\
&=1-2p \\
&\quad \text{since } p=\overline{p}+\widehat{p}-3~\overline{p} \widehat{p}.  
\end{align*}
Therefore, in both cases $RA(\{x_1,x_2\})=1-2p$ which completes the proof. 
\end{proof}
\begin{corollary}
For any rooted binary phylogenetic ultrametric tree and the $N_3$-model, the Fitch algorithm using all terminal taxa is more accurate, or at least as accurate, for ancestral state reconstruction than using any two terminal taxa, that is
$$RA(X) \geq 1-2p.$$
\end{corollary}

\bibliography{references}

\end{document}